\documentclass[12pt]{article}
\usepackage {palatino}
\usepackage[OT1]{fontenc} 

\usepackage{amsfonts,amsmath,amsthm,amssymb,verbatim,nicefrac,mathtools}
\usepackage[linesnumbered,ruled,vlined]{algorithm2e}

\SetCommentSty{mycommfont}
\usepackage{aliascnt}

\usepackage[pdftex,letterpaper,margin=1in]{geometry} 
\usepackage{afterpage} 
\usepackage[round,longnamesfirst]{natbib}
\usepackage{graphicx, color}
\usepackage{setspace}
\usepackage[colorlinks,pagebackref=false]{hyperref}
\usepackage{epsf}
\usepackage{latexsym, stmaryrd, bm}
\usepackage[titletoc,title]{appendix}
\usepackage{titlesec}
\usepackage{marginnote}
\usepackage{graphicx, subfigure, ftnxtra}
\usepackage{todonotes}
\usepackage{enumitem}
\setlist[enumerate]{topsep=0pt,itemsep=0pt} 
\setlist[itemize]{topsep=0pt,itemsep=0pt} 

\definecolor{dark-red}{rgb}{0.60,0.15,0.15}
\definecolor{dark-blue}{rgb}{0.15,0.15,0.65}
\definecolor{light-blue}{rgb}{.2,1,1}
\definecolor{medium-blue}{rgb}{0,0,0.5}
\definecolor{dark-green}{rgb}{0,.5,0}

\hypersetup{
    pdftitle={Lemonade from Lemons},
    pdfauthor={Kartik and Zhong},
    citecolor={dark-blue},
    bookmarksnumbered=true,
    urlcolor=dark-red,
    linkcolor=dark-red,}

\newenvironment{theoremast}[1]
  {\renewcommand{\thetheoremst}{\ref*{#1}$^\ast$}
   \begin{theoremst}}
  {\end{theoremst}}

\theoremstyle{plain}
\newtheorem{theorem}{Theorem}
\newtheorem{theoremst}{Theorem}

\newtheorem{lemma}{Lemma}

\newtheorem{proposition}{Proposition}

\newtheorem{condition}{Condition}
\newtheorem{corollary}{Corollary}

\theoremstyle{definition}
\newtheorem{definition}{Definition}

\theoremstyle{remark}
\newtheorem{remark}{Remark}

\newcommand{\citepos}[1]{\citeauthor*{#1}'s (\citeyear{#1})}
\newcommand{\citeauthorpos}[1]{\citeauthor*{#1}'s}
\makeatletter
  \renewcommand\@seccntformat[1]{\csname the#1\endcsname.{\hskip.7em\relax}} 
\makeatother

\titlespacing\section{0pt}{10pt plus 2pt minus 2pt}{4pt plus 2pt minus 2pt} 
\titlespacing\subsection{0pt}{6pt plus 2pt minus 2pt}{2pt plus 2pt minus 2pt} 
\titlespacing\subsubsection{0pt}{6pt plus 2pt minus 2pt}{0pt plus 2pt minus 2pt} 
\titlespacing\paragraph{0pt}{8pt plus 2pt minus 2pt}{8pt plus 2pt minus 2pt} 

\newcommand{\overbar}[1]{\mkern 1.5mu\overline{\mkern-1.5mu#1\mkern-1.5mu}\mkern 1.5mu}

\renewcommand{\epsilon}{\varepsilon}
\renewcommand{\phi}{\varphi}
\renewcommand{\bar}{\overline}
\DeclareMathOperator*{\Supp}{Supp}

\DeclareMathOperator*{\argmax}{arg\,max}

\newcommand{\mailto}[1]{\href{mailto:#1}{\texttt{#1}}}
\def\bi{\begin{itemize}}
\def\ei{\end{itemize}}

\def\E{\mathbb{E}}

\def\Reals{\mathbb{R}}

\def\TP2{\text{TP}_2}

\def\Pset{\mathbf{\Pi}}
\def\piall{\ubar{\pi}_s}
\def\pius{\ubar{\pi}_s^{us}}
\def\pifb{\ubar{\pi}_s^{fb}}
\def\d{\mathrm{d}}

\let\oldfootnote\footnote
\renewcommand\footnote[1]{\oldfootnote{\hspace{.4mm}#1}}

\makeatletter
\renewenvironment{proof}[1][\proofname] {\par\pushQED{\qed}\normalfont\topsep6\p@\@plus6\p@\relax\trivlist\item[\hskip\labelsep\bfseries#1\@addpunct{.}]\ignorespaces}{\popQED\endtrivlist\@endpefalse}
\makeatother

\let\oldFootnote\footnote
\newcommand\nextToken\relax

\renewcommand\footnote[1]{%
    \oldFootnote{#1}\futurelet\nextToken\isFootnote}

\newcommand\isFootnote{%
    \ifx\footnote\nextToken\textsuperscript{,}\fi}

\newcommand{\ubar}[1]{\mkern 1.5mu\underline{\mkern-1.5mu#1\mkern-1.5mu}\mkern 1.5mu}

\begin{document}

\onehalfspacing 

\begin{titlepage}

\title{Lemonade from Lemons:\\ Information Design and Adverse Selection\footnote{We thank Laura Doval, David Kim, Wenhao Li, Pietro Ortoleva, Alessandro Pavan, and anonymous referees for useful comments. We also received helpful feedback from a number of audiences. Tianhao Liu and Yangfan Zhou provided excellent research assistance.}}
\date{\today}
\author{Navin Kartik\thanks{Department of Economics, Yale University. Email: \mailto{nkartik@gmail.com}. The author was affiliated with Columbia University during most of the work on this paper.} 
\and Weijie Zhong\thanks{Graduate School of Business, Stanford University. Email: \mailto{weijie.zhong@stanford.edu}.}}

\date{\today\\[5pt] {\small (First Draft: December 2, 2019)}}
\maketitle

\begin{abstract}
	A seller posts a price for a single object. The seller's and buyer's values may be interdependent. We characterize the set of payoff vectors across all information \mbox{structures}. Simple feasibility and individual-rationality constraints identify the \mbox{payoff} set. The buyer can obtain the entire surplus; often, non-informational mechanisms cannot enlarge the payoff set. We also study payoffs when the buyer is more informed than the seller, and when the buyer is fully informed. All three payoff sets coincide (only) in notable special cases---in particular, when there is complete breakdown in a ``lemons market'' with an uninformed seller and fully-informed buyer.
\end{abstract}

\bigskip

\thispagestyle{empty}
\end{titlepage}

\section{Introduction}

\paragraph{Motivation.}
Asymmetric information affects market outcomes, both in terms of efficiency and distribution. For example, adverse selection can generate dramatic market failure \citep{akerlof1970} or skew wages in labor markets \citep{greenwald86}, while consumers can secure information rents from a monopolist \citep{MussaRosen78}. Much of the work in information economics prior to the last decade took the market participants' information as given and studied properties of a particular market structure or mechanism, or tackled these properties across various mechanisms.

Our paper joins a recent wave of research---elaborated subsequently---by instead asking: what is the scope for different market outcomes as the participants' information varies? We are motivated by the fact that in the digital age, the nature of information that sellers (e.g., Amazon) have about consumers is ever-changing. Consumers and regulators do have some control over this information, of course. In some cases, it is plausible that a seller's information is a subset of the consumer's. But in other cases, the seller may well know \emph{more} about the consumer's value for a product, or at least have some information the consumer herself does not. This is especially relevant for products the consumer is not already familiar with. Indeed, numerous firms make tailored recommendations to consumers about the products they carry. With social media and other sources of information diffusion, the possible correlation in information across two sides of a market seems truly limitless. It is this variety of possible information that our paper focuses on.

Our paper fixes a simple, canonical market mechanism and studies the possible market outcomes across a variety of information structures, including \emph{all} of them. We model two parties, Buyer and Seller, who can trade a single object. Buyer's value for the object is a random $v\in [\underline v,\overline v]\subset \Reals$. Seller's cost of providing the object, or equivalently, her value from not trading, is $c(v)\leq v$. Thus, values may be interdependent, but trade is always efficient.\footnote{We describe here our baseline model presented in \autoref{sec:model}. \autoref{sec:discussion} discusses extensions, including cases when Buyer's value does not pin down Seller's cost and when trade is not always efficient.} The environment, i.e., the function $c(\cdot)$ and the distribution of $v$, is commonly known. Seller posts a price $p\in \Reals$, and Buyer decides whether to buy.

This stylized setting subsumes a variety of possibilities, depending on the shape of the cost function $c(\cdot)$ and the parties' information about the value $v$. With an informed Buyer and an uninformed Seller, there is adverse selection when $c(\cdot)$ is increasing, while there is favorable or advantageous selection when $c(\cdot)$ is decreasing.\footnote{\citet{Jovanovic82} uses the term `favorable selection'. \citet{EinavFinkelstein11} use `advantageous', and discuss both adverse and advantageous selection in the context of insurance markets, with references to empirical evidence on both.} If, on the other hand, Seller is better informed than Buyer, signaling becomes relevant; the price can serve as a credible signal if the two parties' information is suitably correlated \citep[e.g.,][]{BagwellRiordan91}. A constant $c(\cdot)$ captures an environment in which there is no uncertainty about Seller's cost; this is the canonical monopoly pricing problem when Seller is uninformed about $v$, and third-degree price discrimination when Seller has some partial information while Buyer is better informed.

\paragraph{Summary of results.} For any given environment (i.e., Seller's cost function and the distribution of Buyer's values), we seek to identify the possible market outcomes. Specifically, we  are interested in the ex-ante expected payoffs that obtain, given sequentially rational behavior, in an equilibrium under \emph{some} information structure.\footnote{As detailed in \autoref{sec:model}, an information structure specifies a joint distribution of private signals for each party conditional on the value $v$. This induces an extensive-form game of incomplete information. Our primary solution concept is weak Perfect Bayesian equilibrium; we also address refinements for our constructive arguments.} We provide three results, each of which covers a different class of information structures. Our main theorems are Theorems \ref{thm:joint}/\ref{thm:joint:discrete}, which impose no restrictions on information, and \autoref{thm:seller}, which applies when Buyer is better informed than Seller in the sense of \citet{Blackwell53}; in fact, \autoref{thm:seller} applies more broadly, as elaborated later. \autoref{thm:buyer} concerns a fully-informed Buyer who knows his value $v$.  We view each of these three cases as intellectually salient and economically relevant. Plainly, these payoff sets must be ordered by set inclusion: \autoref{thm:joint}'s is the largest; \autoref{thm:seller}'s is intermediate; and \autoref{thm:buyer}'s the smallest. Figure \ref{fig:illustration} below summarizes.

\begin{figure}[htbp]
	\centering

\tikzset{every picture/.style={line width=0.75pt}} 

\begin{tikzpicture}[x=0.75pt,y=0.75pt,yscale=-1,xscale=1]

\draw    (168.03,39.75) -- (190.59,39.75) ;
\draw [shift={(192.59,39.75)}, rotate = 180] [fill={rgb, 255:red, 0; green, 0; blue, 0 }  ][line width=0.08]  [draw opacity=0] (7.2,-1.8) -- (0,0) -- (7.2,1.8) -- cycle    ;
\draw  [draw opacity=0][fill={rgb, 255:red, 0; green, 0; blue, 255 }  ,fill opacity=0.5 ] (168.03,39.75) -- (296.82,169.46) -- (168.03,169.46) -- cycle ;
\draw    (168.03,259.35) -- (399.95,259.35) ;
\draw [shift={(401.95,259.35)}, rotate = 180] [fill={rgb, 255:red, 0; green, 0; blue, 0 }  ][line width=0.08]  [draw opacity=0] (12,-3) -- (0,0) -- (12,3) -- cycle    ;
\draw    (168.03,259.35) -- (168.03,22.85) ;
\draw [shift={(168.03,20.85)}, rotate = 90] [fill={rgb, 255:red, 0; green, 0; blue, 0 }  ][line width=0.08]  [draw opacity=0] (12,-3) -- (0,0) -- (12,3) -- cycle    ;
\draw [color={rgb, 255:red, 0; green, 149; blue, 255 }  ,draw opacity=1 ]   (168.03,39.75) -- (385.47,259.35) ;
\draw    (186.28,169.46) .. controls (192.2,158.26) and (183.05,159.25) .. (225.72,159.47) ;
\draw [shift={(227.7,159.48)}, rotate = 180.25] [fill={rgb, 255:red, 0; green, 0; blue, 0 }  ][line width=0.08]  [draw opacity=0] (7.2,-1.8) -- (0,0) -- (7.2,1.8) -- cycle    ;
\draw    (309.1,181.75) -- (351.67,181.75) ;
\draw [shift={(353.67,181.75)}, rotate = 180] [fill={rgb, 255:red, 0; green, 0; blue, 0 }  ][line width=0.08]  [draw opacity=0] (7.2,-1.8) -- (0,0) -- (7.2,1.8) -- cycle    ;
\draw  [draw opacity=0][fill={rgb, 255:red, 0; green, 42; blue, 255 }  ,fill opacity=0.28 ] (333.67,206.31) -- (167.86,206.31) -- (168.03,169.46) -- (296.82,169.46) -- cycle ;
\draw  [draw opacity=0][fill={rgb, 255:red, 0; green, 42; blue, 255 }  ,fill opacity=0.11 ] (358.23,230.87) -- (167.86,230.87) -- (167.86,206.31) -- (333.67,206.31) -- cycle ;
\draw  [draw opacity=0][fill={rgb, 255:red, 0; green, 0; blue, 255 }  ,fill opacity=0.5 ] (295.85,57.17) -- (308.14,69.45) -- (295.85,69.45) -- cycle ;
\draw  [draw opacity=0][fill={rgb, 255:red, 0; green, 42; blue, 255 }  ,fill opacity=0.28 ] (308.14,94.02) -- (295.85,94.02) -- (295.85,81.74) -- cycle ;
\draw  [draw opacity=0][fill={rgb, 255:red, 0; green, 42; blue, 255 }  ,fill opacity=0.11 ] (295.85,130.86) -- (295.85,118.58) -- (308.14,130.86) -- cycle ;
\draw  [draw opacity=0][fill={rgb, 255:red, 251; green, 0; blue, 0 }  ,fill opacity=1 ] (184.21,169.46) .. controls (184.21,168.36) and (185.14,167.47) .. (186.28,167.47) .. controls (187.43,167.47) and (188.35,168.36) .. (188.35,169.46) .. controls (188.35,170.57) and (187.43,171.46) .. (186.28,171.46) .. controls (185.14,171.46) and (184.21,170.57) .. (184.21,169.46) -- cycle ;
\draw  [draw opacity=0][fill={rgb, 255:red, 251; green, 0; blue, 0 }  ,fill opacity=1 ] (165.96,39.75) .. controls (165.96,38.64) and (166.88,37.75) .. (168.03,37.75) .. controls (169.17,37.75) and (170.1,38.64) .. (170.1,39.75) .. controls (170.1,40.85) and (169.17,41.74) .. (168.03,41.74) .. controls (166.88,41.74) and (165.96,40.85) .. (165.96,39.75) -- cycle ;

\draw (148.58,13.63) node [anchor=north west][inner sep=0.75pt]  [font=\footnotesize]  {$\pi _{s}$};
\draw (404.65,252.29) node [anchor=north west][inner sep=0.75pt]  [font=\footnotesize]  {$\pi _{b}$};
\draw (163.09,266.8) node [anchor=north west][inner sep=0.75pt]  [font=\footnotesize]  {$0$};
\draw (227.13,150.3) node [anchor=north west][inner sep=0.75pt]  [font=\footnotesize] [align=left] {Akerlof};
\draw (312.19,57.9) node [anchor=north west][inner sep=0.75pt]  [font=\footnotesize] [align=left] {Fully-Informed Buyer};
\draw (310.71,80.9) node [anchor=north west][inner sep=0.75pt]  [font=\footnotesize] [align=left] {Uninformed Seller\\(More Informed Buyer)};
\draw (311.57,119.44) node [anchor=north west][inner sep=0.75pt]  [font=\footnotesize] [align=left] {All info. Structures};
\draw (37.68,224.03) node [anchor=north west][inner sep=0.75pt]  [font=\footnotesize]  {$\max\{\underline{v} -\mathbb{E}[ c( v)] ,0\}$};
\draw (154.87,34.88) node [anchor=north west][inner sep=0.75pt]  [font=\scriptsize] [align=left] {A};
\draw (155.36,165.07) node [anchor=north west][inner sep=0.75pt]  [font=\scriptsize] [align=left] {B};
\draw (155.36,202.41) node [anchor=north west][inner sep=0.75pt]  [font=\scriptsize] [align=left] {D};
\draw (155.36,225.99) node [anchor=north west][inner sep=0.75pt]  [font=\scriptsize] [align=left] {F};
\draw (303.72,163.6) node [anchor=north west][inner sep=0.75pt]  [font=\scriptsize] [align=left] {C};
\draw (340.08,201.91) node [anchor=north west][inner sep=0.75pt]  [font=\scriptsize] [align=left] {E};
\draw (365.62,224.51) node [anchor=north west][inner sep=0.75pt]  [font=\scriptsize] [align=left] {G};
\draw (261.98,58.04) node [anchor=north west][inner sep=0.75pt]  [font=\footnotesize] [align=left] {ABC};
\draw (261.27,82.96) node [anchor=north west][inner sep=0.75pt]  [font=\footnotesize] [align=left] {ADE};
\draw (261.4,120.15) node [anchor=north west][inner sep=0.75pt]  [font=\footnotesize] [align=left] {AFG};
\draw (194.39,30.9) node [anchor=north west][inner sep=0.75pt]  [font=\footnotesize] [align=left] {Full or No Information};
\draw (355,167.1) node [anchor=north west][inner sep=0.75pt]  [font=\footnotesize] [align=left] {Expected Surplus\\Frontier};

\end{tikzpicture}
	\caption{Outcome under different restrictions on information structures}
	\label{fig:illustration}
\end{figure}
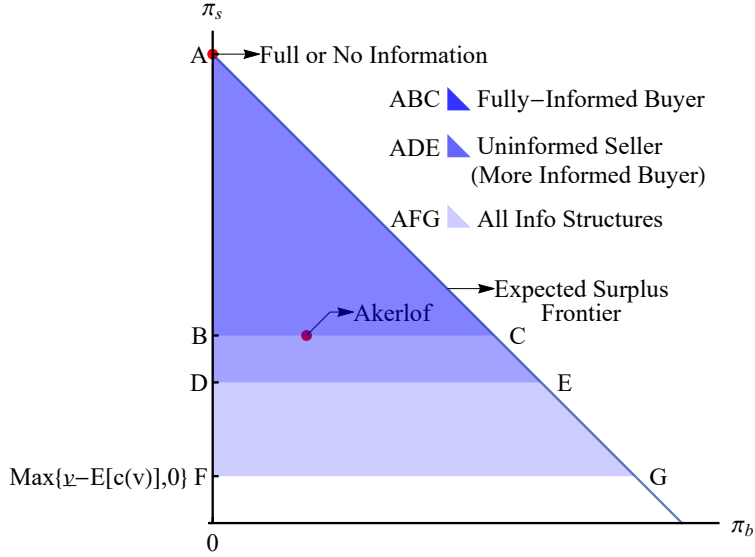

In the figure's axes, $\pi_b$ and $\pi_s$ represent respectively Buyer's and Seller's ex-ante expected utilities or payoffs (for readability, we often drop the ``expected'' qualifier). The no-trade payoffs are normalized to zero. The three triangles, $AFG$, $ADE$, and $ABC$, depict Theorems \ref{thm:joint}--\ref{thm:buyer} respectively. That payoffs must lie within $AFG$ is straightforward: Buyer can guarantee himself a payoff of zero by not purchasing; Seller can guarantee herself not only a payoff of zero (by posting any price $p>\overline v$, which will not be accepted) but also $\underline v-\E[c(v)]$ (by pricing at or just below $\underline v$, which will be accepted); and the sum of payoffs cannot exceed the trading surplus $\E[v-c(v)]$. We refer to the first two constraints as \emph{individual rationality} and the third as \emph{feasibility}.

\autoref{thm:joint} says that {every} feasible and individually rational payoff pair can be implemented, i.e., obtains in an equilibrium under some information structure. It is immediate that point $A$ obtains when both parties learn $v$ (full information) or neither party has any information (no information). More interestingly, at the point $G$ trade occurs with probability one and Buyer obtains the entire surplus, despite Seller posting the price. While perhaps surprising, this outcome obtains with sparse information structures. For simplicity, suppose $\E[c(v)|v>\underline v]\leq \max\{\underline v,\E[c(v)]\}$. Then Buyer can be uninformed while Seller learns whether $v=\underline v$ or $v>\underline v$. In equilibrium, Seller prices at $p=\max\{\underline v,\E[c(v)]\}$ regardless of her signal and Buyer purchases. If Seller were to deviate to a higher price, Buyer would reject because he believes $v=\underline v$. \autoref{sec:joint} explains how a single information structure in fact implements every point in the triangle $AFG$. \autoref{thm:joint:discrete} there discusses how a richer information structure using imperfectly-correlated signals ensures implementation in \citepos{KW82} sequential equilibrium in discretized versions of the model.

Turning to \autoref{fig:illustration}'s triangle $ADE$, \autoref{thm:seller} establishes that the payoff pair in any equilibrium when Buyer is better informed than Seller arises in an equilibrium of an(other) information structure in which Seller is uninformed.\footnote{We stipulate that a better-informed Buyer does not update his value from Seller's price, even off the equilibrium path, in line with the ``no signaling what you don't know'' requirement \citep{FT91} that is standard in versions of Perfect Bayesian Equilibrium and implied by sequential equilibrium.} In other words, there is no loss of generality in studying an uninformed Seller so long as Buyer is better informed. When $c(\cdot)$ is increasing, such information generates a game with adverse selection; when $c(\cdot)$ is decreasing there is favorable selection. Seller's payoff along the line segment $DE$ is the lowest payoff she can get in any information structure in which she is uninformed.\footnote{It is because Seller cannot commit to the price as a function of her signal that she can be harmed (i.e., receive a payoff lower than that on the $DE$ segment) with more information. However, \autoref{thm:seller} assures that Seller is not harmed so long as Buyer is better informed.} \autoref{thm:seller} further establishes that any point within the $ADE$ triangle can be implemented with some such information structure by suitably varying Buyer's information. In fact, we show that higher slices of the triangle (i.e., those corresponding to larger Seller's payoff) can always be implemented by reducing Buyer's information in the sense of \citet{Blackwell53}. We also explain in \autoref{sec:seller} why the triangle $ADE$ actually characterizes all payoffs that can obtain when Buyer does not update from Seller's price, even if Buyer is not better informed than Seller.

Finally, the $ABC$ triangle in \autoref{fig:illustration} depicts \autoref{thm:buyer}, which characterizes all payoff pairs when Buyer is fully informed, i.e., learns $v$. We use the term ``Akerlof'' to describe a fully-informed Buyer and an uninformed Seller, as these information structures are standard in the adverse-selection literature; the corresponding payoff pair is marked as such in the figure. Depending on the environment's primitives, the Akerlof point can be anywhere on the segment $BC$, including at the extreme points. Any feasible payoff pair that satisfies Buyer's individual rationality and gives Seller at least her Akerlof payoff can be implemented with a fully-informed Buyer by suitably varying Seller's information.

An implication of Theorems \ref{thm:joint}--\ref{thm:buyer} is that it is without loss, in terms of ex-ante equilibrium payoffs, to focus on information structures in which Buyer is fully informed if and only if Seller's Akerlof payoff coincides with her individual rationality constraint. This coincidence occurs if and only if the Akerlof market can have full trade (Seller prices at $p=\underline v$ and gets payoff $\underline v-\E[c(v)]\geq 0$) or no trade (the price is $p \geq \overline v$ and both parties' payoffs are $0$). As detailed in \autoref{r:FBsufficient} of \autoref{sec:buyer}, in all other cases the point $B$ in \autoref{fig:illustration} is distinct from the point $D$ (and hence also $F$), which means that Buyer can obtain a higher payoff with less-than-full information, while keeping Seller uninformed.\footnote{This substantially broadens the message of \citet{RS17} on the benefits to restricting Buyer's information.} Furthermore, under a reasonable condition, if Seller's individual rationality constraint is zero (i.e., $\underline v\leq \E[c(v)]$), then point $D$ is also distinct from $F$; see \autoref{r:USsufficient} in \autoref{sec:seller}. When $D$ and $F$ are distinct, maximizing Buyer's payoff, i.e., achieving point $G$, requires Seller to have some information Buyer does not and an equilibrium with {price-dependent beliefs}: after conditioning on his signal, Buyer must update about $v$ from the price either on or off the equilibrium path.

While our results do not speak directly to the economics of privacy, recently reviewed by \citet{ATW16}, they do offer a notable twist. Consumer welfare can be higher when a monopolist has information about a consumer's valuation that the consumer does not; indeed, maximizing consumer welfare in our single-unit setting frequently necessitates that. This is an implication of our Theorems \ref{thm:joint}/\ref{thm:joint:discrete} and \autoref{thm:seller}. 

Relatedly, we view those results as cautioning against assuming away the possibility that sellers have information buyers do not---not only does this seem relevant in practice, as discussed earlier, but it has significant payoff consequences. An alternative perspective on \autoref{thm:joint}/\autoref{thm:joint:discrete} is that they are negative results: ``anything goes'' without restrictions on information structures or equilibria. From that perspective, \autoref{thm:seller} reveals that what is essential to restrain payoffs to its smaller set (triangle $ADE$ in \autoref{fig:illustration}) is that prices do not provide Buyer with information. Whether this is a consequence of Buyer being better informed than Seller or a principle of equilibrium selection does not matter. \autoref{thm:buyer} characterizes the additional payoff restrictions that obtain from Buyer being fully informed.

\paragraph{Related literature.}Our questions and results are most closely related to \cite*{BBM15}, \cite{RS17}, and \citet[Section 4]{MakrisRenou18}. These papers---only the relevant section of the third paper---study the monopoly pricing problem in which there is uncertainty only about Buyer's valuation. This is the special case of our interdependent-values model with a constant function $c(v) \leq \underline v$.\footnote{Related to \citet{RS17} are also \citet{Du18} and \citet{LM19}, who consider worst-case profit guarantees for Seller in static and dynamic environments, respectively. \citet{TW19} qualify \citet{RS17} by allowing Seller to supply Buyer with additional information, although Seller cannot  have any private information of her own.}
We study interdependent values because of its importance in many economic environments; substantively and methodologically, we explore whether and how insights from the monopoly-pricing problem hold more generally.

\cite*{BBM15} assume Buyer is fully informed, and hence can only vary Seller's information.\footnote{Less directly related to our work, there are also recent papers extending the approach of \citet{BBM15} to monopolistic markets with multiple products \citep[e.g.,][]{Ichihashi20,HS22,TV24}, oligopolistic markets with differentiated products \citep[e.g.,][]{EGKL22,EGKL23,BBM24}, and profit-maximizing information design by intermediaries \cite[e.g.,][]{Yang22}.} Our \autoref{thm:buyer}, corresponding to triangle $ABC$ in \autoref{fig:illustration}, is a generalization of their main result to our environment; the key step in our methodology is to construct the ``isoprofit distributions'', which reduces to their ``extremal markets'' in monopoly pricing. An economic lesson from our analysis is that unlike in monopoly pricing, there are salient interdependent-value environments in which a  fully-informed Buyer can achieve all implementable payoffs, even when the Akerlof market has inefficiency; see \autoref{r:FBsufficient}.

\citet{RS17} assume Seller is uninformed and only vary Buyer's information. For the monopoly-pricing environment, they derive one part of our \autoref{thm:seller}, viz., they identify the triangle $ADE$ in \autoref{fig:illustration} as the implementable set when Seller is uninformed. Even for this result, our methodology is quite different from theirs because we do not assume a linear $c(\cdot)$ function; our methodology delivers new insights, including that noted in \autoref{r:garbling}, and also sheds a different light on \citeauthor{RS17}' payoff triangle.
When we specialize to a linear $c(\cdot)$, we can obtain a sharper characterization of the point $E$, which extends \citeauthor{RS17}' characterization of Buyer-optimal information to an interdependent-values environment; see \autoref{thm:seller_linear}. 

Our Theorems \ref{thm:joint}--\ref{thm:buyer} establish that the ``alignment'' principle highlighted by \citet{BBM24}---Buyer surplus/payoff is maximized when total surplus (the sum of Buyer and Seller payoffs) is maximized---extends with a single seller beyond the settings of \citet{BBM15} and \citet{RS17}, both in terms of the information structures considered and to interdependent values. However, we qualify this point in \autoref{ssec:negative:surplus} when trade may be inefficient, as was also illustrated by example in \citet[][online appendix]{RS17}.

While our main interest is in interdependent values, our results provide new insights even for monopoly pricing. \autoref{thm:seller} implies that the \citet{RS17} bounds are without loss so long as Buyer is better informed than Seller; or, more generally, in equilibria in which Buyer's belief is price independent after conditioning on his own signal. On the other hand, \autoref{thm:joint} establishes that any feasible and individually rational payoff pair can be implemented absent these restrictions: in particular, Buyer may even get all the surplus. This latter point has a parallel with \citet{MakrisRenou18}. As an application of their general results on ``revelation principles'' for information design in multi-stage games, \citepos{MakrisRenou18} Proposition 1 deduces an analog of our \autoref{thm:joint} for the (independent values) monopoly pricing problem. We share with \citeauthor{MakrisRenou18} an emphasis on sequential rationality;\footnote{\citeauthor{MakrisRenou18} use an apparatus of ``sequential Bayes correlated equilibrium'', which we do not. In our approach, note that Seller's individual rationality constraint described earlier hinges, when $\E[c(v)]<\underline v$, on Buyer's behavior being sequentially rational even off the equilibrium path.} we go further by establishing in \autoref{thm:joint:discrete} off-the-equilibrium-path belief consistency in the sense of sequential equilibrium \citep{KW82}.  We also show in \autoref{thm:joint}/\ref{thm:joint:discrete} that a single information structure implements all payoffs in the relevant triangle.

Other authors have studied different aspects of more specific changes of information in adverse-selection settings, maintaining that one side of the market is better informed than the other. \citet{Levin01} identifies conditions under which the volume of trade decreases when one party is kept uninformed and the other's information become more effective in the sense of \citet{Lehmann98}; see also \citet{Kessler01}. Assuming a linear payoff structure, \citet{BIJL21} consider how certain changes in Gaussian information affect the volume of trade, surplus, and a certain quantification of adverse selection. 

\citet{Dang08}, \citet{PavanTirole23}, and \citet{Thereze23} study endogenous costly information acquisition with interdependent values, using different assumptions about the nature and timing of information acquisition and the underlying economic environment. By contrast, we do not have strategic or costly information acquisition; rather, the informational environment is exogenously (and costlessly) varied.

In a model with interdependent values where they hold fixed a partially-informed buyer's information, \citet{DPR24} characterize the outcome---including what information the seller should have---that maximizes the buyer's payoff. They highlight that the solution typically involves the seller being partially informed. \citet{Garcia18} solve for socially optimal information provision in an insurance setting with adverse selection; owing to a cross-subsidization motive, full information disclosure is typically not optimal. \citet{PollrichStrausz23} study a third-party certifier in an adverse-selection environment. Their environment corresponds to our Buyer being fully informed and facing a competitive market of sellers. Among other things, they discuss implementable payoff vectors for Buyer (conditional on type) and their certifier.

\par
The rest of our paper proceeds as follows. We introduce our model, equilibrium concept(s), and certain classes of information structures in \autoref{sec:model}. \autoref{sec:results} presents the main results: implementable payoffs when the information structure is arbitrary or varies within canonical classes. \autoref{sec:discussion} contains discussion and extensions. All formal proofs are in the \hyperref[appendices]{Appendices}.

\section{Model}
\label{sec:model}
\subsection{Primitives}
\label{ssec:model:primitives}
There are two players, Seller and Buyer; given the assumptions that follow, Buyer can be viewed as representing a market of buyers. Seller may sell an indivisible good to Buyer. Buyer's value for the good is $v\in V\subset\Reals$, where $V$ is a compact (finite or infinite) set with $\ubar{v}\equiv \min V<\max V  \equiv \overbar{v}$. The value $v$ is drawn from a probability measure $\mu$ with support $V$. Seller's cost of production is given by a function $c(v)$.  We assume $c:V\to \Reals$ is continuous, $v- c(v)\ge 0$ for all $v$, and $\E[v-c(v)]>0$. Hence, the trading surplus is nonnegative for all $v$ and positive for a positive measure of $v$. (Throughout, expectations are with respect to the prior measure $\mu$ unless indicated otherwise; `positive' means `strictly positive' and similarly elsewhere.) Note that the function $c(v)$ need not be monotonic. \autoref{sec:discussion} extends the model to Seller's cost being stochastic even conditional on $v$, and considers the possibility of negative trading surplus. We call $\Gamma\equiv (c,\mu)$ an \emph{environment}. We refer to an environment with a constant $c(\cdot)$ function as that of \emph{monopoly pricing}.

An information structure consists of signal spaces for each party and a joint signal distribution. (We abuse terminology and refer to `distribution' even though `measure' would sometimes be more precise.) Formally, there is a probability space $(\Omega, \mathcal{F}, P)$, complete and separable metric spaces $T_s$ and $T_b$ (equipped with their Borel sigma algebras), and an integrable function $X:\Omega\to T_s\times T_b\times V$. 
We hereafter suppress the probability space and define, with an abuse of notation, $P(D)=P(X^{-1}(D))$ for any measurable $D\subset T_s\times T_b\times V$.\footnote{We write $\subset$ for ``weak subset''.} Each realization of random variable $X$ is a triplet $(t_b,t_s,v)$, where $t_b\in T_b$ is Buyer's  signal and $t_s\in T_s$ is Seller's signal. For $i\in\{s,b,v\}$, let $P_i$ denote the corresponding marginal distribution of $P$ on dimension $T_i$, with the convention $T_v\equiv V$. We require $P_v= \mu$; this is an iterated expectation or ``Bayes plausibility'' requirement. Denote an information structure by $\tau$.

The environment $\Gamma$ and information structure $\tau$ define the following \emph{game}:
\begin{enumerate}
	\item The random variables $(t_b,t_s,v)$ are realized. Signal $t_b$ is privately observed by Buyer and signal $t_s$ privately observed by Seller. Neither party observes $v$.
	\item Seller posts a price $p\in \Reals$. 
	\item Buyer accepts or rejects the price. If Buyer accepts, his von-Neumann Morgenstern payoff is $v-p$ and Seller's is $p-c(v)$. If Buyer rejects, both parties' payoffs are normalized to $0$.
\end{enumerate}

Note that because the signal spaces $T_b$ and $T_s$ are abstract and the two parties' signals can be arbitrarily correlated conditional on $v$, there is no loss of generality in assuming that each party privately observes their own signal. For example, public information can be captured by perfectly correlating (components of) $t_b$ and $t_s$.

We highlight that our notion of an information structure involves parties receiving information only at the outset. A more permissive notion would also allow Buyer to receive information after Seller posts her price, as in the literature on multi-stage information design \citep{MakrisRenou18,DovalEly20}. Permitting that would not change some of our results, in particular Theorems \ref{thm:joint}/\ref{thm:joint:discrete} and \autoref{thm:buyer}, but would expand the implementable set characterized in \autoref{thm:seller}. Methodologically, our interest in only ex-ante information means that existing ``revelation principles'' do not directly apply. 

Our assumption that Seller simply posts a price---rather than using more complicated mechanisms---is not restrictive for our main results.  \autoref{rem:prices} elaborates later.

\subsection{Strategies and Equilibria}
\label{ssec:PBE}
In the game defined by $(\Gamma,\tau)$, denote Seller's strategy by $\sigma$ and Buyer's by $\alpha$. Following \cite{milgrom1985distributional}, we define $\sigma$ as a distributional strategy: $\sigma$ is a joint distribution on $\mathbb{R}\times T_s$ whose marginal distribution on $T_s$ must be the Seller's signal distribution. So $\sigma(\cdot|t_s)$ is Seller's price distribution given her signal $t_s$.\footnote{Here $\sigma(\cdot|t_s)$ is the regular conditional distribution, which exists and is unique almost everywhere because $T_s$ is a standard Borel space \citep[pp.~229--230]{Durrett95}. Similarly for subsequent such notation; we drop ``almost everywhere'' qualifiers unless essential.} Buyer's strategy $\alpha: \mathbb{R}\times T_b\to [0,1]$ maps each price-signal pair $(p,t_b)$ into a trading probability. A strategy profile $(\sigma,\alpha)$ induces expected utilities for Buyer and Seller $(\pi_b,\pi_s)$ in the natural way:
\begin{align*}
		\pi_b&=\int (v-p)\alpha(t_b,p)\sigma(\d p| t_s)P(\d t_s,\d t_b,\d v),\\
		\pi_s&=\int (p-c(v))\alpha(t_b,p)\sigma(\d p| t_s)P(\d t_s,\d t_b,\d v).
\end{align*}

Our baseline equilibrium concept is weak Perfect Bayesian equilibrium. Since Seller's action is not preceded by Buyer's we can dispense with specifying beliefs for Seller. For Buyer, it suffices to focus on his belief about the value $v$ given his signal and the price; we denote this distribution by $\nu(v|p,t_b)$.
\begin{definition}
	\label{defi:PBE}
	A strategy profile $(\sigma,\alpha)$ and beliefs $\nu(v|p,t_b)$ is a \emph{weak perfect Bayesian equilibrium} (wPBE) of game $(\Gamma,\tau)$ if:
		\begin{enumerate}
		\item Buyer plays optimally at every information set given his belief:
			$$\alpha(p,t_b)=
			\begin{cases}
				1&\text{if }\E_{\nu(v|p,t_b)}[v]> p\\
				0&\text{if }\E_{\nu(v|p,t_b)}[v]< p;
			\end{cases}
			$$
		\item Seller plays optimally:
			$$\sigma\in \argmax_{\hat \sigma}\int (p-c(v))\alpha(p,t_b)\hat \sigma(\d p| t_s)P(\d t_s,\d t_b,\d v);$$
		\item Beliefs satisfy Bayes rule on path: for every measurable $D \subset \mathbb{R}\times T_s\times T_b \times V$,
			$$ \int_{D} \nu(\d v|p,t_b) \sigma(\d p|t_s)P(\d t_s,\d t_b, V)=\int_D \sigma(\d p|t_s)P(\d t_s,\d t_b,\d v).$$
	\end{enumerate}
\end{definition}

We have formulated Seller's optimality requirement ex ante, but Buyer's at each information set. The latter is needed to capture sequential rationality. The former is for (notational) convenience; this choice is inconsequential because Seller moves before Buyer.

Hereafter, ``equilibrium'' without qualification refers to a wPBE. As is well understood, wPBE permits significant latitude in beliefs off the equilibrium path. We will subsequently discuss refinements.

\subsection{Implementable Payoffs and Canonical Information Structures}
\label{ssec:info:set}
We now define the set of \emph{implementable} equilibrium outcomes---that is, the payoffs that obtain in some equilibrium under some information structure---and some canonical classes of information structures.

For a game $(\Gamma,\tau)$, let the equilibrium payoff set be
$$\Pi(\Gamma,\tau)\equiv \left\{ (\pi_b,\pi_s):\exists\text{ wPBE of $(\Gamma,\tau)$ with payoffs $(\pi_b,\pi_s)$} \right\}.$$ 
Denote the class of all information structures by $\mathbf{T}$ and define $$\Pset(\Gamma)\equiv \bigcup_{\tau\in \mathbf{T}}\Pi(\Gamma,\tau).$$ 
That is, for environment $\Gamma$, $\Pset(\Gamma)$ is the set of all equilibrium payoff pairs that obtain under some information structure.

\paragraph{Uninformed Seller.} An information structure has \emph{uninformed Seller} if $T_s$ is a singleton: Seller's own signal contains no information about Buyer's value $v$, and hence neither about Seller's cost $c(v)$.\footnote{Among all reasonable notions of uninformed Seller (e.g., one might only require Seller to have no information about $\E[v]$, while permitting information about $c(v)$), we take the most restrictive one. Our results will imply that a more permissive notion would not change the relevant implementable sets---in particular, that characterized in \autoref{thm:seller}. This point also applies to our notion of more-informed Buyer defined shortly.} When discussing such information structures, we write the associated distribution as just $P(t_b,v)$ and Seller's strategy as just $\sigma(p)$, omitting the argument $t_s$ in both cases. The class of all uninformed-Seller information structures is denoted $\mathbf{T}_{us}$. 

\paragraph{Fully-informed Buyer.} An information structure has \emph{fully-informed Buyer} if Buyer's signal fully reveals his value $v$. Formally, this holds if $T_b=V$ and the conditional distribution on $V$, $P(\cdot|t_b)$, satisfies $P(\{t_b\}|t_b)=1$. We denote the class of fully-informed-Buyer information structures by $\mathbf{T}_{fb}$. Note that a fully-informed Buyer need not know Seller's signal; but that is irrelevant to Buyer, because his optimal action after any price only depends on his known value.

\paragraph{More-informed Buyer.} An information structure has \emph{more-informed Buyer} if Buyer has more information than Seller. Formally, this holds when $v$ and $t_s$ are independent conditional on $t_b$, i.e., for any measurable $D_s\subset T_s$ and $D_v \subset V$, $P(D_s\times D_v|t_b)=P(D_s|t_b)P(D_v|t_b)$. Another way to interpret this requirement is that random variable $t_b$ must be statistically sufficient for $t_s$ with respect to $v$, i.e., $t_b$ is more informative than $t_s$ about $v$ in the sense of \citet{Blackwell53}. We denote the class of more-informed-Buyer information structures by $\mathbf{T}_{mb}$. Naturally, information structures with uninformed Seller or fully-informed Buyer are cases of more-informed Buyer: both $\mathbf{T}_{us}$ and $\mathbf{T}_{fb}$ are subclasses of $\mathbf{T}_{mb}$.

\paragraph{No updating from price.}
For more-informed-Buyer information structures, it is desirable to impose further requirements on Buyer's equilibrium belief. Since Seller's price can only depend on her own signal, and this signal contains no additional information about $v$ given Buyer's signal, the price is statistically uninformative about $v$ given Buyer's signal. Consequently, Buyer's posterior belief should be price independent once his signal has been conditioned upon. Formally, regardless of the price $p$, the equilibrium belief $\nu(\cdot|p,t_b)$ must satisfy $$\int_D \nu(\d v|p,t_b)P(\d t_b,T_s,V)=\int_D P(\d t_b,T_s,\d v)$$ for any measurable $D\subset T_b\times V$. We refer to this condition as \emph{price-independent beliefs}.\footnote{The condition is distinct from ``passive beliefs'', which is typically used to restrict beliefs after off-the-equilibrium-path events.} Note that although we have motivated the condition by Buyer being more informed than Seller, the condition is meaningful even otherwise, capturing the notion of equilibria in which there is no signaling by Seller, or, more precisely, that Buyer does not learn anything about his value $v$ from the price that he does not already learn from his own signal.  In a more-informed-Buyer information structure, price-independent beliefs would be implied by the ``no signaling what you don't know'' requirement \citep{FT91} frequently imposed in versions of perfect Bayesian equilibrium, and the concept of sequential equilibrium \citep{KW82} in finite versions of our setting.\footnote{An example clarifying our terminology may be helpful. If both Seller and Buyer are fully informed of $v$, then the natural equilibrium---the unique sequential equilibrium in a finite version of the game---has Seller pricing at $p=v$ and Buyer's belief being degenerate on $v$ regardless of Seller's price. This equilibrium has price-independent beliefs, even though Seller's price and Buyer's belief are perfectly correlated ex ante. The point is that Buyer's belief does not depend on price conditional on his signal.}
 
Some more notation will be helpful. Define
\begin{align*}
	&\Pi^*(\Gamma,\tau)\equiv \left\{ (\pi_b,\pi_s):\exists\text{ wPBE of $(\Gamma,\tau)$ with price-independent beliefs and payoffs} (\pi_b,\pi_s) \right\},\\[5pt]
	&\Pset^*(\Gamma)\equiv \bigcup_{\tau\in \mathbf{T}}\Pi^*(\Gamma,\tau), \text{ and }\\[5pt]
	&\Pset^*_i(\Gamma)\equiv \bigcup_{\tau\in \mathbf{T}_i}\Pi^*(\Gamma,\tau) \text{ for } i=us,fb,mb.
\end{align*}
So $\Pi^*$ and $\Pset^*$ are analogous to the implementable payoff sets $\Pi$ and $\Pset$ defined earlier, but restricted to equilibria with price-independent beliefs. $\Pset^*_{us}$, $\Pset^*_{fb}$ and $\Pset^*_{mb}$ are the implementable payoff sets when further restricted to uninformed-Seller, fully-informed-Buyer, and more-informed-Buyer information structures. Plainly, for any environment $\Gamma$, $$\Pset^*_{us}(\Gamma) \cup \Pset^*_{fb}(\Gamma) \subset \Pset^*_{mb}(\Gamma) \subset \Pset^*(\Gamma) \subset \Pset(\Gamma).$$

\section{Main Results}
\label{sec:results}
Our goal is to characterize equilibrium payoff pairs across information structures in an arbitrary environment $\Gamma$. In particular, we seek to characterize the five sets $\Pset(\Gamma)$, $\Pset^*(\Gamma)$, $\Pset^*_{mb}(\Gamma)$, $\Pset^*_{us}(\Gamma)$, and $\Pset^*_{fb}(\Gamma)$. Let $$S(\Gamma)\equiv \E[v-c(v)]$$ be the (expected) surplus from trade in environment $\Gamma$. This quantity will play an important role.

\subsection{All Information Structures}
\label{sec:joint}

Define Seller's payoff guarantee as
\begin{align*}
\piall(\Gamma)\equiv \max\left\{  \ubar{v}-\E\left[c(v) \right],0 \right\}.
\end{align*}
To interpret this quantity, observe that it is optimal for Buyer to accept the price $\ubar v$ no matter his belief. Therefore, Seller can guarantee herself the (expected) profit $\underline v-\E[c(v)]$ no matter what the information structure is. More precisely, she can guarantee $v-\E[c(v)]-\epsilon$ for any $\varepsilon>0$, since sequential rationality requires Buyer to accept any price $\ubar v-\epsilon$. Similarly, Seller can also guarantee zero profit offering a price $p>\overbar{v}$. Hence Seller's payoff in any equilibrium with any information structure must be at least $\piall(\Gamma)$.

On the other hand, Buyer can guarantee himself the payoff $\pi_b=0$ by rejecting all prices. It follows that the implementable set $\Pset(\Gamma)$ must satisfy three simple constraints: (1) Seller's ``individual rationality'' constraint $\pi_s\ge \piall(\Gamma)$; (2) Buyer's ``individual rationality'' constraint $\pi_b\ge 0$; and (3) the feasibility constraint $\pi_b+\pi_s\le S(\Gamma)$. 

Our first result is that these individual rationality and feasibility constraints are also sufficient for a payoff pair to be implementable.

\begin{theorem}
\label{thm:joint}
The set of implementable outcomes under all information structures and equilibria is
\begin{align*}
	\Pset(\Gamma)=
\left\{ 
			\begin{array}{ll}
				&\pi_b\ge0\\
				(\pi_b,\pi_s):&\pi_s\ge \piall(\Gamma)\\
											&\pi_b+\pi_s\le S(\Gamma) 
			\end{array}
		\right\}.
\end{align*}
\end{theorem}
\begin{proof}
See \autoref{appendix:B}.
\end{proof}
\autoref{thm:joint} says that the set $\Pset(\Gamma)$ corresponds to the triangle $AFG$ in \autoref{fig:illustration}. In particular, Buyer can receive the entire surplus beyond Seller's payoff guarantee. This is perhaps surprising, as Seller has substantial bargaining power. Note that when $\ubar v\le \E[c(v)]$, a reasonable condition, Seller's payoff guarantee is zero; in that case, \autoref{thm:joint} implies that Buyer can obtain the entire surplus.\footnote{In monopoly pricing with $c(\cdot)=\ubar v$, Seller's payoff guarantee of zero is lower than the revenue guarantee identified by \citet[Section 5]{Du18}, which is typically positive. \citeauthorpos{Du18} notion is different from ours.}

The proof of \autoref{thm:joint} is in fact straightforward. Suppose, for expositional simplicity, $\underline v\geq \E[c(v)]$. Fix the trivial information structure in which neither player receives any information and consider the following family of strategy profiles. Seller randomizes between two prices, some $p_l\in[\ubar{v},\E[v]]$ and $p_h=\E[v]$, with probability $\sigma(p_l)\in[0,1]$. Buyer accepts $p_l$ with probability one and accepts $p_h$ with probability $\alpha(p_h)$, where $\alpha(p_h)\in[0,1]$ is specified to make Seller indifferent between the two prices. That is, $\alpha(p_h)(p_h-\E[c(v)])=p_l-\E[c(v)]$. The expected payoffs from this strategy profile are
$$		\pi_b=\sigma(p_l)(\E[v]-p_l) \text{ and }  \pi_s=p_l-\E[c(v)].$$

As $p_l$ traverses the interval $[\ubar{v},\E[v]]$, Seller's payoff $\pi_s$ traverses $[\piall(\Gamma),S(\Gamma)]$. Given any $p_l$, Buyer's payoff $\pi_b$ traverses $[0,S(\Gamma)-\pi_s]$ as $\sigma(p_l)$ traverses $[0,1]$. Therefore, the proposed strategy profiles induce all the payoff pairs stated in \autoref{thm:joint}.

We are left to specify beliefs $\nu$ for Buyer. After prices $p_l$ and $p_h$ Buyer holds the prior belief $\mu$. After any other (necessarily off-path) price Buyer's belief is that $v=\ubar{v}$, and so Buyer rejects all prices $p\in [\underline v,\infty)\setminus \{p_l,p_h\}$. It is straightforward to confirm that the specified $(\sigma,\alpha,\nu)$ constitute a wPBE.

To get more insight into the construction above, consider its implication for monopoly pricing with $c(\cdot)=\ubar v$. The equilibrium with $p_l=\ubar v$ (hence $\alpha(p_h)=0$, i.e., the buyer rejects the higher price) and $\sigma(p_l)=1$ corresponds to the monopolist deterministically pricing at $\ubar v$ and Buyer purchasing. Given that both sides of the market receive no information, why doesn't the monopolist deviate to any price in $(\ubar v,\E[v])$? The reason is that in this equilibrium, the consumer will then not buy because he updates his belief to $v=\ubar v$. Such updating is compatible with wPBE because the equilibrium concept places no restrictions on off-path beliefs. This may seem like a game-theoretic misdirection: Buyer's beliefs are not consistent with ``no signaling what you don't know''. Put differently, since we have a (weakly) more-informed Buyer information structure, we ought to impose the price-independent beliefs condition described in \autoref{ssec:info:set}; that would imply Buyer must purchase at any price $p<\E[v]$.

But the message of \autoref{thm:joint} does not rely on the permissiveness of wPBE. To illustrate, continue with the above monopoly-pricing environment, and suppose $\underline v$ has positive prior probability. Consider Buyer remaining uninformed but Seller learning whether $v=\underline v$ or $v>\underline v$. Now Buyer's off-path belief that $v=\underline v$ is consistent with ``no signaling what you don't know''. More generally, using richer information structures, we can prove that any payoff pair identified in \autoref{thm:joint} can be approximately implemented as a \emph{sequential equilibrium} \citep{KW82} in a suitably discretized game.

\begin{theoremast}{thm:joint}		
\label{thm:joint:discrete}
		Fix any $\varepsilon>0$. There is $\Delta>0$ such that for any finite price grid with size $\Delta$, there is a finite information structure inducing a game with a set of sequential equilibrium payoffs that is an $\varepsilon$-net of $\Pset(\Gamma)$, the set of implementable outcomes under all information structures and equilibria.\footnote{An information structure is finite if the signal spaces $T_b$ and $T_s$ are finite. A finite price grid of size $\Delta$ means that the set of prices is finite, with minimum price no higher than $\underline v$ and maximum price no lower than $\overline v$, and any two consecutive prices are no more than $\Delta$ apart. Sequential equilibrium is defined in the obvious way for the ``induced'' finite game where Nature directly draws $(t_b,t_s)$, rather than first drawing $v$, and players' payoffs from trading are defined directly as $\left(\E[v|t_b,t_s]-p,p-\E[c(v)|t_b,t_s]\right)$.\par	For $Y\subset \Reals^2$ and $\epsilon>0$, the set $A\subset Y$ is an $\epsilon$-net of $Y$ if for each $y\in Y$ there is $a\in A$ such that $\lVert y-a\rVert < \epsilon$, where $\lVert \cdot \rVert$ is the Euclidean distance.}
\end{theoremast}
\begin{proof}
See \autoref{appendix:thm:1:discrete}.
\end{proof}
In fact, the proof of \autoref{thm:joint:discrete} establishes even more: the sequential equilibria in the discretized games satisfy a natural version of the D1 refinement \citep{CK87}. We relegate the  logic to the Appendix, but mention here that we use imperfectly-correlated signals for Buyer and Seller.\footnote{The idea behind D1 is to ask, for any off-path price, whether 
one type of Seller would deviate for any Buyer mixed response that another type would. Our construction has multiple Buyer types that are imperfectly correlated with Seller types. So different types of Seller have different beliefs about Buyer types. This blunts the power of dominance considerations, to the point where D1 does not exclude any Seller type from the support of Buyer's off-path belief.}

Even when our environment is specialized to monopoly pricing, it is worth highlighting two contrasts between \autoref{thm:joint}/\ref{thm:joint:discrete} and results of \citet*{BBM15} and \cite{RS17}. First, we find that by not restricting the monopolist to be uninformed, the implementable payoff set typically expands rather dramatically: trade can be efficient with the monopolist securing none of the surplus beyond her payoff guarantee, $\piall$, which may be zero. (\citeauthor{RS17} establish, implicitly, that the implementable set with an uninformed monopolist is a superset of \citeauthorpos{BBM15}, where the consumer is fully informed.) We will see in \autoref{sec:seller} that what is crucial to this expansion is price-\emph{dependent} beliefs. In particular, the proof of \autoref{thm:joint:discrete} uses an information structure in which Buyer is not better informed than Seller --- if he were, then sequential equilibrium would imply price-independent beliefs. Second, \autoref{thm:joint:discrete} establishes that for a given $\varepsilon>0$, a single information structure (and price grid) can be used to approximate the entire payoff set $\Pset(\Gamma)$, analogously to the construction described after \autoref{thm:joint} that used a single information structure.\footnote{In fact, if one lets the price grid vary with $\epsilon$, then a single information structure implements exactly, rather than approximately, in sequential equilibrium all payoffs $\epsilon$-away from the boundary of $\Pset(\Gamma)$.  See \autoref{prop:joint:discrete} in the Appendix for a formal statement.} \citet*{BBM15} and \cite{RS17}, on the other hand, vary information structures to span their payoff sets.

\subsection{More-informed Buyer and Price-independent Beliefs}
\label{sec:seller}

In some economic settings it is plausible that Buyer is more informed than Seller. How does a restriction to such information structures, i.e., $\tau \in \mathbf{T}_{mb}$, affect the implementable payoff set? It turns out that what is in fact crucial is price-independent beliefs. We have explained earlier why it is desirable to impose this condition when Buyer is more informed than Seller, but that the condition is well defined even otherwise. If Buyer is not more informed than Seller, then price-independent beliefs ought to be viewed as an equilibrium restriction. Readers should be bear in mind that, to reduce repetition, the qualifier ``with price-independent beliefs'' applies for the rest of this subsection unless stated explicitly otherwise.

It is useful to define 
\begin{equation}
	\pius(\Gamma)\equiv \inf\left\{ \pi_s:\exists (\pi_b,\pi_s)\in \Pset^*_{us}(\Gamma) \right\}\label{e:ubar_pi}	
\end{equation}
as the infimum payoff that an \emph{uninformed} Seller can obtain, no matter Buyer's information (among equilibria with price-independent beliefs, we stress). Plainly, $\pius(\Gamma)\geq \piall(\Gamma)$. In monopoly pricing with $V=[\underline v,\overline v]$ and $c(\cdot)=\underline v$, \citepos{RS17} characterization of the consumer-optimal information structure identifies $\pius$, establishing that $\pius>\piall$. If there is no trade due to adverse selection when Seller is uninformed and Buyer has some information, then $\pius=\piall=0$. We do not have a general explicit formula for $\pius$; \autoref{sec:linear} provides it for linear $c(\cdot)$. Nonetheless, we establish next that (i) the only additional restriction on equilibrium payoffs imposed by price-independent beliefs is a lower bound of $\pius$ for Seller, and (ii) uninformed-Seller information structures implement all such payoffs.

\begin{theorem}
\label{thm:seller}
The set of implementable outcomes under all information structures in equilibria with price-independent beliefs is
the same as the set of implementable outcomes under uninformed-Seller information structures in equilibria with price-independent beliefs. Moreoever:
\begin{enumerate}
\item \label{seller1}$\Pset^*(\Gamma)=\Pset^*_{mb}(\Gamma)=\Pset^*_{us}(\Gamma)$.
\item \label{seller2}$\Pset^*_{us}(\Gamma)=\{(\pi_b,\pi_s)\in \Pset(\Gamma) : \pi_s\geq \pius(\Gamma)\}$.
\item \label{seller3}For any $(\pi_b,\pi_s)\in \Pset^*_{us}(\Gamma)$ with $\pi_s>\pius(\Gamma)$, there is $\tau\in \mathbf{T}_{us}$ with $\Pi(\Gamma,\tau)=\left\{ (\pi_b,\pi_s) \right\}$.
\end{enumerate}
\end{theorem}
\begin{proof}
See \autoref{sec:proof_seller}.
\end{proof}
\begin{remark}
We believe the substance of \autoref{thm:seller} would hold using discretizations and sequential equilibria, analogous to \autoref{thm:joint:discrete}. As previously noted, sequential equilibrium implies price-independent beliefs when Buyer is more informed than Seller.
\end{remark}

To digest \autoref{thm:seller}, note that $\Pset^*(\Gamma)\supset \Pset^*_{mb}(\Gamma)\supset \Pset^*_{us}(\Gamma)$ is trivial. So part \ref{seller1} of the theorem amounts to establishing the reverse inclusions. The intuition for those---given part \ref{seller2}'s characterization of $\Pset^*_{us}$---is fairly straightforward: with price-independent beliefs, additional information cannot harm Seller, even though it could alter the set of equilibria. So Seller's lowest payoff obtains when she is uninformed.

The characterization in part \ref{seller2} of payoffs with an uninformed Seller corresponds to the triangle $ADE$ in \autoref{fig:illustration}. Part \ref{seller3} of the theorem assures ``unique implementation'' of all implementable payoffs satisfying $\pi_s>\pius(\Gamma)$. That is, for any such payoff pair, there is an uninformed-Seller information structure such that all equilibria (with price-independent beliefs) induce exactly that payoff pair. Unique implementation is appealing for multiple reasons, one of which is that it obviates concerns about which among multiple payoff-distinct equilibria is more reasonable.

Let us describe how we obtain the characterization of $\Pset^*_{us}(\Gamma)$ and  unique implementation. There are two steps. The first ensures that there is some information structure, call it $\tau^* \in \mathbf T_{us}$, that implements Seller's payoff $\pius(\Gamma)$. That is, we ensure that the infimum in \eqref{e:ubar_pi} is in fact a minimum.\footnote{The difficulty is in establishing suitable continuity. Uninformed-Seller information structures can be viewed as probability measures over Buyer's beliefs, with convergence in the sense of the weak* topology. This topology ensures continuity, with respect to probability measures, of expectations of continuous or at least Lipschitz (and bounded) functions. However, Seller's expected payoff is not the expectation of a Lipschitz function, as Seller's profit is truncated at the price she charges.} While this argument is technical, knowing $\tau^*$ exists is useful in what follows. The second, and economically insightful, step is to construct information structures that implement every point in the triangle $\Pset^*_{us}(\Gamma)$ by suitably garbling the information structure $\tau^*$. The construction is illustrated in \autoref{fig:thm:seller}. Consider the distribution of Buyer's posterior mean of his valuation $v$ in information structure $\tau^*$. (Given price-independent beliefs, Buyer's posterior mean is a sufficient statistic for his decision.) For simplicity, suppose this posterior-mean distribution has a density, as depicted by the red curve in \autoref{fig:thm:seller}. Fix any $(\pi_b,\pi_s) \in \Pset^*_{us}(\Gamma)$.

First, there is some number $z^*$ such that $\pi_b+\pi_s$ is the total surplus from trading only when Buyer's posterior mean is greater than $z^*$. Next, there is some price $p^*\geq z^*$ such that Seller's payoff is $\pi_s$ if all these trades were to occur at price $p^*$.\footnote{That $p^*\geq z^*$ follows from $\pi_s\geq \pius(\Gamma)$, as $\pius(\Gamma)$ itself is weakly larger than Seller's payoff from posting price $z^*$ (and thus trading with the same set of Buyer posterior means) under information structure $\tau^*$.}  Note that $p^*$ must be no larger than the expected Buyer posterior mean conditional on that being above $z^*$, for otherwise $\pi_b < 0$. We claim that the information structure $\tau^*$ can be garbled so that $p^*$ is an equilibrium price and trade occurs only when Buyer's posterior mean is greater than $z^*$.
		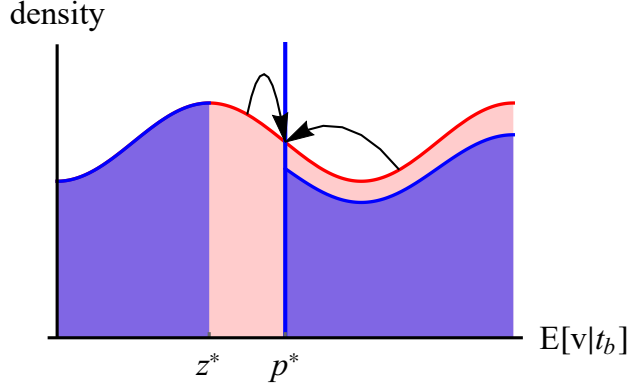
\begin{figure}[htbp]
		\centering

\tikzset{every picture/.style={line width=0.75pt}} 

\begin{tikzpicture}[x=0.75pt,y=0.75pt,yscale=-1,xscale=1]

\draw  [draw opacity=0][fill={rgb, 255:red, 255; green, 0; blue, 0 }  ,fill opacity=0.28 ] (192.98,142.92) -- (205.7,142.44) -- (220.88,146.15) -- (232.98,151.62) -- (257.18,163.38) -- (274.62,170.31) -- (294.31,174.01) -- (310.72,172.73) -- (326.1,168.05) -- (338.61,161.61) -- (359.95,146.39) -- (359.95,237.97) -- (147.61,238.23) -- (147.85,168.94) -- (156.47,167.73) -- (164.26,163.22) -- (174.11,154.2) -- (182.93,146.63) -- cycle ;
\draw  [draw opacity=0][fill={rgb, 255:red, 0; green, 0; blue, 255 }  ,fill opacity=0.35 ] (199.95,142.2) -- (199.95,238.26) -- (147.61,238.23) -- (147.85,168.94) -- (153.6,168.67) -- (162.01,164.83) -- (169.8,158.39) -- (178.62,149.85) -- (184.16,145.99) -- (192.98,142.92) -- cycle ;
\draw    (147.61,238.23) -- (380.02,238.23) ;
\draw [shift={(382.02,238.23)}, rotate = 180] [fill={rgb, 255:red, 0; green, 0; blue, 0 }  ][line width=0.08]  [draw opacity=0] (12,-3) -- (0,0) -- (12,3) -- cycle    ;
\draw    (147.61,238.23) -- (147.61,113.59) ;
\draw [shift={(147.61,111.59)}, rotate = 90] [fill={rgb, 255:red, 0; green, 0; blue, 0 }  ][line width=0.08]  [draw opacity=0] (12,-3) -- (0,0) -- (12,3) -- cycle    ;
\draw [color={rgb, 255:red, 255; green, 0; blue, 0 }  ,draw opacity=1 ]   (199.95,142.2) .. controls (245.08,140.27) and (282.41,210.5) .. (359.95,146.39) ;
\draw [color={rgb, 255:red, 0; green, 0; blue, 255 }  ,draw opacity=1 ]   (147.85,168.94) .. controls (171.24,167.65) and (173.29,142.2) .. (199.95,142.2) ;
\draw  [draw opacity=0][fill={rgb, 255:red, 0; green, 0; blue, 255 }  ,fill opacity=0.35 ] (251.85,171.76) -- (273.8,179.17) -- (295.95,182.23) -- (318.92,180.14) -- (333.9,174.98) -- (345.18,169.34) -- (359.95,158.23) -- (359.95,237.97) -- (229.9,237.97) -- (229.9,161.29) -- cycle ;
\draw [color={rgb, 255:red, 0; green, 0; blue, 255 }  ,draw opacity=1 ]   (230.31,161.53) .. controls (254.93,172.16) and (304.77,203.17) .. (360.15,158.39) ;
\draw [color={rgb, 255:red, 0; green, 0; blue, 255 }  ,draw opacity=1 ]   (229.9,114.33) -- (229.9,237.97) ;
\draw    (212.67,143.11) .. controls (216.09,128.88) and (223.03,132.87) .. (228.17,147.08) ;
\draw [shift={(228.81,148.91)}, rotate = 251.66] [fill={rgb, 255:red, 0; green, 0; blue, 0 }  ][line width=0.08]  [draw opacity=0] (7.2,-1.8) -- (0,0) -- (7.2,1.8) -- cycle    ;
\draw    (308.67,171.68) .. controls (278.11,145.16) and (252.81,129.36) .. (232.51,148.15) ;
\draw [shift={(231.27,149.34)}, rotate = 314.85] [fill={rgb, 255:red, 0; green, 0; blue, 0 }  ][line width=0.08]  [draw opacity=0] (7.2,-1.8) -- (0,0) -- (7.2,1.8) -- cycle    ;

\draw (125.97,96.68) node [anchor=north west][inner sep=0.75pt]  [font=\small] [align=left] {density};
\draw (192,242) node [anchor=north west][inner sep=0.75pt]  [font=\small]  {$z^*$};
\draw (224,242) node [anchor=north west][inner sep=0.75pt]  [font=\small]  {$p^*$};
\draw (390.8,229.8) node [anchor=north west][inner sep=0.75pt]  [font=\small]  {$\mathbb{E}[ v|t_{b}]$};

\end{tikzpicture}
		\caption{Construction of garbling of $\tau^*$}
		\label{fig:thm:seller}
	\end{figure}
	The garbling is illustrated in \autoref{fig:thm:seller} as the distribution depicted by the blue curve and line. There is one signal that Buyer receives when the original posterior mean is between $z^*$ and $p^*$, and also receives with some probability when the original posterior mean is above $p^*$. The probability is chosen to make the posterior mean from this signal exactly $p^*$. Apart from this one new signal, Buyer receives the original signal in $\tau^*$. Plainly, this is a garbling of $\tau^*$ and hence is feasible. 
	
	\autoref{fig:thm:seller} makes clear why the new information structure has an equilibrium with price $p^*$ and Buyer breaking indifference in favor of trading: (i) Seller's profit from posting any price below $z^*$ is the same as under $\tau^*$ and hence no larger than $\pius(\Gamma)$; (ii) similarly, Seller's profit from posting any price above $p^*$ is no higher than some fraction of $\pius(\Gamma)$; and (iii) any price between $z^*$ and $p^*$ is worse that price $p^*$. Moreover, since Seller's profit from offering any price other than $p^*$ is no more than $\pius(\Gamma)$, it follows that when $\pi_s>\pius(\Gamma)$, Buyer must break indifference as specified for Seller to have an optimal price, and the equilibrium payoffs are unique.

\begin{remark}\label{r:garbling}
The above logic establishes that given any $\tau \in \mathbf T_{us}$ that implements some $(\pi_b,\pi_s)$, $\tau$ can be garbled to uniquely implement any $(\pi'_b,\pi'_s)\in \Pset(\Gamma)$ such that $\pi'_s > \pi_s$. That is, an uninformed Seller's payoff can always be strictly raised, and Buyer's payoff reduced (strictly, so long as it was not already zero), by garbling Buyer's information. Even when specialized to the case of monopoly pricing, this provides a different perspective on why \citet{RS17} obtain a payoff triangle. More importantly, our methodology also handles the case of interdependent values---specifically, a nonlinear cost function $c(\cdot)$.

\end{remark}

\begin{remark}
\label{r:USsufficient}
According to Theorems \ref{thm:joint}/\ref{thm:joint:discrete} and \autoref{thm:seller}, uninformed-Seller information structures cannot implement all implementable payoff pairs in an environment $\Gamma$ if and only if $\pius(\Gamma)>\piall(\Gamma)$. This inequality fails if $\pius(\Gamma)=0$, since that implies $\pius(\Gamma)=\piall(\Gamma)=0$. An example is when there is no trade due to adverse selection when Seller is uninformed and Buyer has some information. On the other hand, $\pius(\Gamma)>\piall(\Gamma)$ if
\begin{equation}
\label{e:USinsufficient}
	\ubar v\leq \E[c(v)] \text{ and } \forall v\in V, \ c(v)<v.
\end{equation}
To see why, notice that in any uninformed-Seller information structure, Seller can price at slightly less than Buyer's highest posterior mean valuation and guarantee trade with only (a neighborhood of) that Buyer type. If $c(v)<v$ for all $v$, this gives Seller a positive expected payoff, and hence $\pius(\Gamma)>0$.\footnote{More precisely: as $V$ is compact, $c(v)<v$ for all $v$ implies there exists $\varepsilon>0$ such that $v-c(v)>\varepsilon$. Given any uninformed-Seller information structure, let $\overbar{m}_v$ be the highest posterior mean valuation in the support of the posterior means induced by Buyer's signals. So there is positive probability of Buyer signals with posterior mean valuations at least $\overbar{m}_v-{\varepsilon}/{2}$. By pricing at $\overbar{m}_v-{\varepsilon}/{2}$, Seller's expected cost conditional on trade is bounded above by $\overbar{m}_v-\varepsilon$, and hence Seller's profit conditional on trade is at least $\epsilon/2>0$. It follows that $\pius>0$.} But the first inequality in \eqref{e:USinsufficient} is equivalent to $\piall(\Gamma)=0$. Hence \eqref{e:USinsufficient} implies $\pius(\Gamma)>\piall(\Gamma)$. We observe that Condition \eqref{e:USinsufficient} is compatible with severe adverse selection resulting in very little trade when Seller is uninformed and Buyer is (partially or fully) informed.
\end{remark}

\subsection{Fully-Informed Buyer}
\label{sec:buyer}

We now turn to the third canonical class of information structures: Buyer is fully informed of his value $v$. As this is a special case of a more-informed Buyer, we maintain price-independent beliefs throughout this subsection. 

Faced with a fully informed Buyer and any sequentially rational Buyer strategy, an uninformed Seller can guarantee the profit level
\begin{align*}
	\pifb(\Gamma)\equiv \sup_{p}\int_p^{\overbar{v}}(p-c(v))\mu(\d v)
\end{align*}
regardless of her information. Plainly, $\pifb(\Gamma) \geq \pius(\Gamma)$. In monopoly pricing with $c(\cdot)=\underline v$, \citet{RS17} have shown that $\pifb>\pius$; if there is no trade due to adverse selection when Buyer is fully informed and Seller is uninformed, then $\pifb=\pius=0$. We establish below that when Buyer is fully informed, $\pifb$ is the only additional constraint on equilibrium payoffs.

\begin{theorem}
\label{thm:buyer}
The set of implementable outcomes under fully-informed--Buyer information structures and equilibria with price-independent beliefs is
$$	\Pset^*_{fb}(\Gamma)= \{(\pi_b,\pi_s)\in \Pset(\Gamma) : \pi_s\geq \pifb(\Gamma)\}.$$
\end{theorem}

\begin{proof}
See \autoref{sec:proof:fb}.
\end{proof}

The payoff set characterized in \autoref{thm:buyer} corresponds to the triangle $ABE$ in \autoref{fig:illustration}. 
Here is the idea behind the result. When Buyer is fully informed, an information structure can be viewed as dividing $v$'s prior distribution, $\mu$, into a set of $\mu_i$ that average to $\mu$, with Seller informed of which $\mu_i$ she faces. \autoref{thm:buyer} is proven by establishing that we can divide $\mu$ suitably so that against each $\mu_i$, Seller is indifferent between pricing at all prices in the support of $\mu_i$, including the price corresponding to $\pifb$ in that environment. {Such a $\mu_i$ is analogous to an ``extremal market'' introduced by \cite{BBM15} in the context of monopoly pricing. To highlight the profit implication of such a distribution and because that implication is relevant across multiple information structures in our paper, we call such a $\mu_i$ an \emph{isoprofit distribution} or IPD.}

\begin{definition}
		\label{defi:IPD}
		$\nu$ is an \emph{isoprofit distribution} (IPD) if
	\begin{align*}
		\int_p^{\overbar{v}}(p-c(s))\nu(\d s)=\text{constant}\ge 0, \ \forall p\in\mathrm{Supp}(\nu).
	\end{align*}
\end{definition}

The Appendix provides a ``greedy'' algorithm to compute IPDs; the algorithm is defined for finite $V$, and we take limits to handle the infinite case. We can sketch how the algorithm works and construct a set of IPDs that average to the prior. Suppose $V=\{v_1,v_2,\ldots,v_K\}$, with $v_i<v_{i+1}$ for $i\in \{1,\ldots,K-1\}$ and $c(v)<v$ for all $v$. Given any small-enough mass of $v_K$, there is a unique mass of type $v_{K-1}$ that makes Seller indifferent between charging price $v_K$ and $v_{K-1}$. (If the mass is too low, Seller prefers $v_K$; if it is too high, she prefers $v_{K-1}$.) Iterating down to keep Seller indifferent between all prices pins down an IPD. Choose the maximum mass of type $v_K$ for which this works. Remove that IPD---i.e., take the conditional distribution after removing the masses of each type according to that IPD--and then repeat the procedure to construct the next IPD. 

Crucially, whenever an IPD is removed, the price corresponding to $\pifb(\Gamma)$ remains optimal in the remaining ``market''; this follows from the IPD's defining property of Seller indifference and an accounting identity. Therefore, Seller's profit in this segmentation of IPDs remains $\pifb(\Gamma)$. Moreover, it is also optimal for Seller to always (i.e., for each $\mu_i$) price so that there is full trade or no trade. Hence, Buyer's expected payoff can be either $0$ or the entire surplus less $\pifb(\Gamma)$. It follows that the fully-informed Buyer information structure defined by this set of IPDs implements point $B$ and $C$ in \autoref{fig:illustration}.  The entire triangle $ABC$ can then be implemented by convexification: randomizing over this information structure (and the two equilibria) and full information (where Seller obtains all the surplus).

\begin{remark}
In the same vein as part \ref{seller3} of \autoref{thm:seller}, one can also establish approximately unique implementation for \autoref{thm:buyer}'s payoff set: for any $(\pi_b,\pi_s) \in \Pset^*_{fb}(\Gamma)$ and any $\varepsilon>0$, there is $\tau\in \mathbf{T}_{fb}$ with $\Pi(\Gamma,\tau)\subset B_{\varepsilon}(\pi_b,\pi_s)$.
\end{remark}

\begin{remark}
\label{r:FBsufficient}
Theorems \ref{thm:joint}--\ref{thm:buyer} imply that fully-informed-Buyer information structures implement all implementable payoff pairs if and only if $\pifb(\Gamma)=\piall(\Gamma)$. In that case, triangles $AFG$ and $ABC$ coincide in \autoref{fig:illustration}. It follows that $\pifb(\Gamma)=\piall(\Gamma)$ only when a fully-informed Buyer and uninformed Seller can result in full trade ($\ubar v\geq \E[c(v)]$ and Seller prices at $\underline v$) or no trade ($\ubar v\leq \E[c(v)]$ and Seller prices at some $p\geq \overline v$).
  Interestingly, when $\pifb(\Gamma)>\piall(\Gamma)$, fully-informed-Buyer information structures cannot even implement all payoff pairs implementable by uninformed-Seller information structures; i.e., triangles $ABC$ and $ADE$ in \autoref{fig:illustration} are distinct if and only if triangles $ABC$ and $AFG$ are distinct. Or to put it another way, when (and only when) $\pifb(\Gamma)>\piall(\Gamma)$ there is an uninformed-Seller information structure that implements some $\pi_s<\pifb(\Gamma)$.\footnote{Pick any $p'>\underline{v}$ such that $p'-\E[c(v)]<\pifb(\Gamma)$. Following the construction described after \autoref{thm:seller}, we can mix all valuations $v\le p'$ with a fraction $\lambda>0$ of valuations $v>p'$ so that the mixture has posterior mean exactly $p'$. The remaining fraction $1-\lambda$ of valuations above $p'$ are revealed to Buyer. With this uninformed-Seller information structure, consider any equilibrium in which Buyer purchases when indifferent. (Such an equilibrium with price-independent belief exists.) Seller's profit is at most $(1-\lambda)\pifb(\Gamma)$ from any price $p>p'$, and $p'-\E[c(v)]$ from price $p=p'$. Hence, Seller's profit is strictly less than $\pifb(\Gamma)$.} Theorems \ref{thm:seller}--\ref{thm:buyer} further imply that this property also characterizes when Buyer can
  benefit from not being fully informed. In the context of monopoly pricing, that can be viewed as characterizing when the buyer can benefit from strategic learning \citep{RS17} rather than market segmentation \citep{BBM15}.
\end{remark}
\begin{remark}
    \label{rem:prices}
Restricting attention to posted prices is
without loss for Theorems \ref{thm:joint}--\ref{thm:buyer}. Since we have implemented all payoffs that are feasible
and individually rational for Seller, allowing Seller to use more complicated mechanisms cannot enlarge the implementable payoff sets. To see why our all our payoffs can still be obtained as well, note that for \autoref{thm:joint}, Seller’s deviation to any other mechanism can simply be deterred by a pessimistic belief. For Theorems \ref{thm:seller}--\ref{thm:buyer}, since Buyer is more informed than Seller, a posted price is
optimal for Seller among all mechanisms \citep{myerson81}.
\end{remark}

\section{Discussion}
\label{sec:discussion}

This section discusses some extensions and refinements of our results.

\subsection{Multidimensionality}
\label{sec:noscalar}
Suppose Buyer and Seller's cost and valuation pair $(c,v)$ is a two-dimensional random variable distributed according to joint distribution $\mu$ with a compact support in $\Reals^2$. The extension of our maintained assumption of commonly known gains from trade is:
for all $(c,v)\in \mathrm{Supp}(\mu)$, $v\ge c$; and $\E[v-c]>0$.
An information structure is now a joint distribution $P(t_b,t_s,c,v)$ whose marginal distribution on $(c,v)$ is $\mu$. 

The substance of Theorems \ref{thm:joint}/\ref{thm:joint:discrete}, \autoref{thm:seller} and \autoref{thm:buyer} still hold.\footnote{A caveat is that in this multidimensional setting, we do not know whether our maintained assumption that Seller only posts a price is without loss---i.e., we do not rule out that certain payoff pairs are not implementable when Seller can use non-posted-price mechanisms (which she might use to improve her payoff). By contrast, in our baseline one-dimensional setting, our results would not be affected if we had allowed Seller to use arbitrary mechanisms. See \cite{che2021robust} and \cite{deb2023multi} for work on information design in multidimensional screening problems.} To see why, let $\underline{v}$ be the lowest valuation in the support of $\mu$. Seller's individual rationality constraint is now $\max\{\underline{v}-\E[c],0\}$, as she can guarantee this profit by setting either a sufficiently high price or a price (arbitrarily close to) $\underline{v}$, regardless of her signal. Abusing notation, we can define a cost function $c(v')\equiv \E_\mu[c|v=v']\le v'$. This results in an environment satisfying all the maintained assumptions of our baseline model, except that $c(\cdot)$ may not be continuous. Such continuity plays no role in proving \autoref{thm:joint} nor \autoref{thm:joint:discrete}. Both \autoref{thm:seller} and \autoref{thm:buyer} use continuity of $c(\cdot)$ to guarantee that $\E_{\nu}[c(v)]$ is a continuous function of $\nu\in\Delta(V)$ for certain convergence arguments. However, in the two-dimensional type environment, $\E_{\nu}[c]$ is still a continuous function of $\nu\in \Delta(C\times V)$. \autoref{thm:buyer} uses upper semi-continuity of Seller's profit in price; boundedness of $c$ and $c\le v$ is sufficient for such upper semi-continuity.

\subsection{Negative Trading Surplus}\label{ssec:negative:surplus}

Returning to our baseline model, we next discuss what happens when trade sometimes generates negative surplus. That is, we drop the assumption that $c(v)\leq v$; we do not require $\E[v-c(v)]>0$ either. Define $S_{\lambda}(\Gamma)$ for $\lambda\in[1,\infty)$ as
\begin{align*}
	S_{\lambda}(\Gamma)\equiv \int_{\underline{v}}^{\overbar{v}}\left[ \underline{v}-c(v)+\lambda(v-\underline{v}) \right]^+\mu(\d v),
\end{align*}
where $[\cdot]^+\equiv \max\left\{ \cdot,0 \right\}$. The function $S_\lambda(\Gamma)$ is a weighted sum of Buyer and Seller payoff assuming that trade occurs at price $p=\underline{v}$ whenever trade creates a positive weighted total payoff, and there is no trade otherwise. It is readily verified that $S_{1}(\Gamma)=\E\left[[v-c(v)]^+\right]$ and $\lim_{\lambda\to \infty}{S_{\lambda}(\Gamma)}/{\lambda}= \E[v]-\underline{v}$. Allowing negative trading surplus does not affect our definition of wPBE. So the notation $\Pset(\Gamma)$ and $\piall(\Gamma)$ still have the same meanings as before. The next proposition shows that $\Pset(\Gamma)$ is now characterized by three constraints: as before, the two individual rationality constraints, $\pi_s\ge\piall(\Gamma)$ and $\pi_b\ge0$; and different now, a Pareto frontier defined by all $S_{\lambda}(\Gamma)$. 
\begin{proposition}
\label{prop:negative:frontier}
Consider all information structures and equilibria when trade can generate negative surplus.
	\begin{align*}
	\Pset(\Gamma)=
\left\{ 
			\begin{array}{ll}
				&\pi_b\ge0\\
				(\pi_b,\pi_s):&\pi_s\ge \piall(\Gamma)\\
				& \lambda\pi_b+\pi_s\le S_{\lambda}(\Gamma),\ \forall \lambda\ge 1 
			\end{array}
		\right\}.	
	\end{align*}
\end{proposition}
\begin{proof}
See \autoref{appendix:negative}.
\end{proof}
\autoref{fig:negative} depicts \autoref{prop:negative:frontier}. The blue triangle's frontier corresponds to total surplus under full trade.\footnote{The figure is drawn assuming $\E[v-c(v)]>0$, which ensures the blue triangle in the figure is nondegenerate. If instead $\E[v-c(v)]\leq 0$, then $\piall(\Gamma)=\max\{\ubar v-\E[c(v)],0\}=0$, and the blue triangle would be the singleton $(0,0)$.}  The union of the blue and red regions is the set $\Pset(\Gamma)$. Each outer blue line has a slope $-\lambda$, with $\lambda\geq 1$, and represents a frontier $\lambda\pi_b+\pi_s=S_{\lambda}(\Gamma)$; the frontier of the red region is defined by their envelope. 

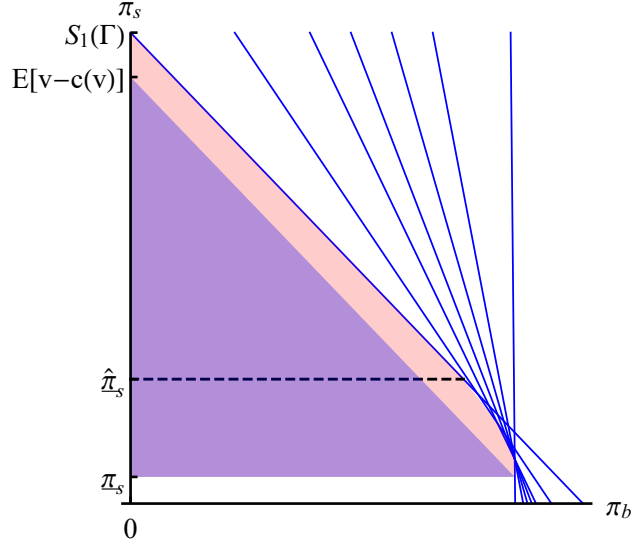
\begin{figure}[htbp]
	\centering

\tikzset{every picture/.style={line width=0.75pt}} 

\begin{tikzpicture}[x=0.75pt,y=0.75pt,yscale=-1,xscale=1]

\draw  [draw opacity=0][fill={rgb, 255:red, 255; green, 0; blue, 0 }  ,fill opacity=0.22 ] (319.63,240.59) -- (321.32,245.65) -- (150.38,245.65) -- (149.82,63.61) -- (303.05,211.07) -- (308.95,218.66) -- (313.45,225.97) -- (317.1,233) -- cycle ;
\draw    (149.41,261.18) -- (361.61,261.18) ;
\draw [shift={(363.61,261.18)}, rotate = 180] [fill={rgb, 255:red, 0; green, 0; blue, 0 }  ][line width=0.08]  [draw opacity=0] (12,-3) -- (0,0) -- (12,3) -- cycle    ;
\draw    (149.41,261.18) -- (149.41,44.8) ;
\draw [shift={(149.41,42.8)}, rotate = 90] [fill={rgb, 255:red, 0; green, 0; blue, 0 }  ][line width=0.08]  [draw opacity=0] (12,-3) -- (0,0) -- (12,3) -- cycle    ;
\draw  [draw opacity=0][fill={rgb, 255:red, 0; green, 0; blue, 255 }  ,fill opacity=0.35 ] (149.26,82.87) -- (321.32,245.65) -- (149.26,245.65) -- cycle ;
\draw  [dash pattern={on 4.5pt off 4.5pt}]  (149.26,211.07) -- (303.05,211.07) ;
\draw [color={rgb, 255:red, 0; green, 0; blue, 255 }  ,draw opacity=1 ][line width=0.75]    (149.82,63.61) -- (354.5,259.99) ;
\draw [color={rgb, 255:red, 0; green, 0; blue, 255 }  ,draw opacity=1 ][line width=0.75]    (213.36,64.87) -- (336.5,261.4) ;
\draw [color={rgb, 255:red, 0; green, 0; blue, 255 }  ,draw opacity=1 ][line width=0.75]    (245.69,64.87) -- (329.76,261.12) ;
\draw [color={rgb, 255:red, 0; green, 0; blue, 255 }  ,draw opacity=1 ][line width=0.75]    (282.52,64.59) -- (324.41,260.55) ;
\draw [color={rgb, 255:red, 0; green, 0; blue, 255 }  ,draw opacity=1 ][line width=0.75]    (321.04,64.59) -- (321.04,260.55) ;

\draw (138.01,30.08) node [anchor=north west][inner sep=0.75pt]  [font=\footnotesize]  {$\pi _{s}$};
\draw (365.44,254.25) node [anchor=north west][inner sep=0.75pt]  [font=\footnotesize]  {$\pi _{b}$};
\draw (144.94,266.78) node [anchor=north west][inner sep=0.75pt]  [font=\footnotesize]  {$0$};
\draw (106.55,58.29) node [anchor=north west][inner sep=0.75pt]  [font=\footnotesize]  {$S_{1}( \Gamma )$};
\draw (75.3,79.25) node [anchor=north west][inner sep=0.75pt]  [font=\footnotesize]  {$\mathbb{E}[ v-c( v)]$};
\draw (122.92,202.33) node [anchor=north west][inner sep=0.75pt]  [font=\footnotesize]  {$\widehat{\underline{\pi }}_{s}$};
\draw (123.17,236.94) node [anchor=north west][inner sep=0.75pt]  [font=\footnotesize]  {$\underline{\pi }_{s}$};

\end{tikzpicture}
	\caption{Outcome when trading surplus can be negative}
\label{fig:negative}
\end{figure}

Let us explain some of the logic underlying \autoref{prop:negative:frontier}/\autoref{fig:negative}. Begin by observing that all the payoffs in the figure's blue triangle can be implemented analogously to our discussion of \autoref{thm:joint}. It is also straightforward that some payoffs outside this set can be implemented. In particular, an information structure that publicly reveals only whether trade is efficient (i.e., whether $v\geq c(v)$ or not) can implement efficient trade with all the surplus accruing to Seller: the point $(0,S_1(\Gamma))$ in \autoref{fig:negative}. Why does maximizing Buyer's payoff now generally require some inefficiency (i.e., why is the red region's frontier not linear when $S_1(\Gamma)>\E[v-c(v)]$)? Consider, for simplicity, $\underline v\geq \E[c(v)]$, so that $\piall(\Gamma)=\underline v-\E[c(v)]$. The bottom-right corner of \autoref{fig:negative}'s blue triangle is then achieved by having trade with probability one at the price $\underline v$, with corresponding Buyer payoff $\E[v]-\underline v$. No higher Buyer payoff is implementable because Seller will never sell at a price below $\underline v$, and subject to that constraint, this outcome maximizes $v-p$ for every $v$. In other words, the maximum implementable Buyer's payoff goes hand in hand with implementing all inefficient trade.

It remains to sketch why the frontier of $\Pset(\Gamma)$ is characterized by the lines defined by $\left\{\lambda \pi_b+\pi_s=S_{\lambda}(\Gamma) \right\}_{\lambda\geq 1}$. 
When type $v$ trades at price $p\geq \underline{v}$, the (ex post) {weighted total payoff} is $ p-c(v)+\lambda(v-p) $, which is weakly below $\underline{v}-c(v)+\lambda(v-\underline{v})$ because $\lambda\ge1$ and $\underline{v}\le p$. When trade does not happen, the weighted total payoff is $0$. Therefore, the weighted total payoff is bounded above by $[\underline{v}-c(v)+\lambda(v-\underline{v})]^+$, and hence each $S_{\lambda}(\Gamma)$ is an upper bound for the expected weighted total payoff. The proof of \autoref{prop:negative:frontier} establishes that each of these upper bounds is tight: for each $\lambda$ there exists an information structure implementing expected weighted total payoff equal to $S_{\lambda}(\Gamma)$. The information structure simply publicly reveals whether the weighted total payoff at price $\underline v$ is negative (a ``negative signal'') or not (a ``positive signal''). For all $v$ such that $\underline{v}-c(v)+\lambda(v-\underline{v})<0$, it holds that $v-c(v)<(1-\lambda)(v-\underline{v})\le0$. Thus, the negative signal creates common knowledge that total surplus is negative, and hence there is no trade. After a positive signal, on the other hand, Seller can be induced to sell at price $\underline{v}$ just as in the discussion of \autoref{thm:joint}. Therefore, the equilibrium expected weighted total payoff is $\E\left[\left[ \underline{v}-c(v)+\lambda(v-\underline{v}) \right]^+\right]=S_{\lambda}(\Gamma)$.

We should note that in certain cases there may not be a tradeoff between maximizing Buyer's payoff and efficiency. Specifically, consider the profit level $\widehat{\underline{\pi}}_s$ that is the maximum of $0$ and Seller's profit from efficient trade at price $\underline v$:
$$\widehat{\underline{\pi}}_s\equiv \left[\int \bm{1}_{v\ge c(v)}(\underline{v}-c(v))\mu(\d v)\right]^+.$$ 
For profit levels above $\widehat{\underline{\pi}}_s$, the situation is analogous to our baseline model once we use a public signal to reveal that trade is efficient, and so the implementable equilibrium payoffs with $\pi_s\ge\widehat{\underline{\pi}}_s$ constitute a triangle, as seen in \autoref{fig:negative}. Hence, if $\widehat{\underline{\pi}}_s=0$, then there is no tradeoff between efficiency and maximizing Buyer's payoff. But when $\widehat{\underline{\pi}}_s>0$, then so long as some trades generate negative surplus, $\widehat{\underline{\pi}}_s>\piall$ and there is a tradeoff.

\subsection{Affine Cost Function}
\label{sec:linear}
Returning to our baseline model, recall that Seller's minimum implementable payoff $\pius(\Gamma)$ under price-independent beliefs (\autoref{thm:seller}) is not amenable to a closed-form formula in general. We now provide such a formula when the cost function $c(v)$ is affine. According to our discussion in \autoref{sec:noscalar}, an affine $c(v)$ subsumes richer environments in which conditional expectations are affine, such as under Gaussian primitives \citep*[cf.~][]{BIJL21}.

\begin{condition}
	$c(v)=\lambda v+\gamma$, for some $\lambda,\gamma \in \Reals$.
	\label{ass:linear}
\end{condition}

Let $F(v)$ be the cumulative distribution function (CDF) corresponding to the prior measure $\mu$. Let $D(\mu)$ be the set of all distributions whose CDF $G$ satisfies
\begin{align*}
		\int_{V}v\d G(v)&=\int_{V}v\d F(v) \quad \text{and} \quad 
		\int_{\underline{v}}^{v}G(s)\d s\le \int_{\underline{v}}^vF(s)\d s, \ \forall v\in V.
\end{align*}
That is, $D(\mu)$ contains all distributions that are mean-preserving contractions (MPC) of $\mu$. It is well known that $D(\mu)$ characterizes the set of distributions of Buyer posterior means that can be generated by any (uninformed-Seller) information structure. We focus on a special family of IPDs (see \autoref{defi:IPD}) whose supports are intervals $[v_{*},v^*]$. Such IPDs have an analytical expression under \autoref{ass:linear}: 
\begin{align}
	G(v)=
	\begin{dcases}
		0 & \text{if $v\leq v_*$}\\
		1 & \text{if $v\geq v^*$}\\
		1-\left( \frac{(1-\lambda)v-\gamma}{(1-\lambda)v_*-\gamma} \right)^{\frac{1}{\lambda-1}}&\text{if $v\in (v_*,v^*)$ and $\lambda\neq 1$}\\[5pt]
		1-e^{\frac{v-v_*}{\gamma}}&\text{if $v\in (v_*,v^*)$ and $\lambda=1$}.
	\end{dcases}
	\label{e:IPDlinear}
\end{align}

$G(v)$ is smooth everywhere except for a mass point at $v^*$. Our maintained assumption that $v\geq c(v)$ implies $(1-\lambda) v\geq \gamma$. Therefore, \autoref{e:IPDlinear} defines an increasing function, i.e., a well-defined CDF. Given any $v_*$, a higher $v^*$ corresponds to increasing $G(v)$ in the sense of first-order stochastic dominance; hence, there is a unique $v^*$ determined by the condition $\E_G[v]=\E_F[v]$. So the family of IPDs is parametrized by a single parameter $v_*$; accordingly, we denote such an IPD by CDF $G_{v_*}$ with density $g_{v_*}$ on $(v_*,v^*)$. The domain for $v_*$ is $[\underline{v},\E[v])$. We separately define $G_{\E[v]}(v)=\mathbf{1}_{v\ge \E[v]}$. \par

It can be verified that a higher $v_*$ lowers the corresponding $v^*$ (i.e., the corresponding intervals $[v_*,v^*]$'s are nested) and that $G_{v_*}(v)$ is pointwise decreasing in $v_*$ within the common support. As a result, for any two different $v_*$'s the corresponding $G_{v_*}$'s cross once. Consequently, all IPDs are ordered according to the MPC order (increasing in $v_*\in[\underline{v},\E[v]]$). Not every $G_{v_*}$ is in $D(\mu)$, but $G_{v_*}\in D(\mu)$ for all $v_*$ larger than some threshold. The following proposition shows that this threshold pins down $\pius$.
\begin{proposition}
\label{thm:seller_linear}
Assume \autoref{ass:linear} and let $p_*\equiv \min\left\{v_*\big|{G}_{v_*}\in D(\mu)  \right\}$. It holds that $\pius(\Gamma)=p_*-\E_{\mu}[c(v)]$.
\end{proposition}
\begin{proof}
See \autoref{appendix:seller:linear}.
\end{proof}
\autoref{thm:seller_linear} states that an uninformed Seller's minimum payoff (with price-independent beliefs) is characterized by a specific IPD of Buyer posterior means. By the Seller-indifference property of IPDs and that $v_*$ is the minimum of $G_{v_*}$'s support, Seller's profit when facing IPD $G_{v_*}$ is $v_*-\E[c(v)]$. \autoref{thm:seller_linear} thus implies that an uninformed Seller's minimum payoff is implemented by the most dispersed IPD that is a MPC of the prior distribution.

In proving \autoref{thm:seller_linear}, the key step is to show that given the prior $\mu$, garbling Buyer's information so that the posterior-mean distribution becomes an IPD makes Seller weakly worse off. As such, it is without loss to only consider IPDs to implement $\pius(\Gamma)$. This makes the problem one dimensional and tractable. To elaborate on the key step, suppose we find a distribution $G(v)$ such that: (i) $G\in D(\mu)$; (ii) $G$ is an IPD; and (iii) there is a $p\in\Supp(G)$ such that $\int_p^{\overbar{v}}G(s)\d s=\int_p^{\overbar{v}}F(s)\d s$. Such a $G$ is the most-dispersed IPD that is a MPC of the prior $\mu$. Consider the following two identities derived using integration by parts:\footnote{$F(p-)$ is defined as the left limit of $F$ at $p$, and similarly $G(p-)$. The integration by parts formula is for Lebesgue-Stieltjes integral.}
\begin{align}
	\lambda  \int_p^{\overbar{v}}F(s)\d s=&-(1-F(p-))\left( p-c(p) \right)+\lambda(\overbar{v}-p)+\int_p^{\overbar{v}}(p-c(s))\d F(s)\label{eqn:linear:1},\\
	\lambda  \int_p^{\overbar{v}}G(s)\d s=&-(1-G(p-))\left( p-c(p) \right)+\lambda(\overbar{v}-p)+\int_p^{\overbar{v}}(p-c(s))\d G(s)\label{eqn:linear:2}.
\end{align}
First, by property (iii) above, the LHS of \autoref{eqn:linear:1} equals the LHS of \autoref{eqn:linear:2}. Second, the MPC condition implies $\int_{v}^{\overbar{v}}(F(s)-G(s))\d s\le0$ for all $v$ and it reaches $0$ when $v=p$, it holds that $F(p-)\le G(p-)$. Therefore, the integral term on RHS of \autoref{eqn:linear:1} must be greater than that of \autoref{eqn:linear:2}. Notice that the integral term on either RHS is Seller's profit when offering price $p$. This implies that the profit from offering $p$ given Buyer mean-valuation distribution $G(v)$ is lower than Seller's maximum profit given $F(v)$. On the other hand, by the IPD property, $p$ is an optimal price given $G(v)$. Therefore the optimal profit under valuation distribution $F$ must be no lower than that under $G$. It follows that to minimize Seller's profit, it is without loss to consider only IPDs within the set $D(\mu)$, which is a one-dimensional subspace.

The logic above generalizes that of  \citet{RS17}. Their monopoly-pricing environment with $c(\cdot)=\underline v$ is covered by \autoref{ass:linear} with $\lambda=0$ and $\gamma=\underline v$. The distribution $G$ in \eqref{e:IPDlinear} then reduces to that identified by \citeauthor{RS17}.\footnote{In recent work, \citet{InostrozaTsoy2025} extend \autoref{thm:seller_linear} to a setting in which Seller does not have all the bargaining power.}

If the prior $\mu$ has binary support, then any cost function $c(v)$ is affine. This case permits an explicit solution for $\pius(\Gamma)$.
\begin{corollary}
		\label{coro:binary}
	Assume $\mu$ has binary support: $V=\left\{ v_1,v_2 \right\}$ with $v_1<v_2$. Let $\lambda=(c(v_2)-c(v_1))/(v_2-v_1)$, and let $p$ be the unique solution to
	\begin{equation}
		(p-c(p))^{\frac{1}{\lambda-1}}(p-\E[c(v)])=(v_2-c(v_2))^{\frac{\lambda}{\lambda-1}}
		\label{eqn:binary}.
	\end{equation}
	It holds that $\pius(\Gamma)=\max\left\{ p,v_1 \right\}-\E[c(v)]$.
\end{corollary}
\begin{proof}
See \autoref{appendix:coro:binary}.
\end{proof}

\vspace{.5in}

\appendix

\numberwithin{equation}{section}
\numberwithin{proposition}{section}
\numberwithin{lemma}{section}
\numberwithin{example}{section}
\numberwithin{claim}{section}

\allowdisplaybreaks

{\Large \noindent \textbf{Appendices}}

\section{Proof of \autoref{thm:joint}}
\label{appendices}
\label{appendix:B}
\begin{proof}
	We first show that $\left\{ (\pi_b,\pi_s):\pi_b\ge0,\pi_s\ge\piall(\Gamma),\pi_b+\pi_s\le S(\Gamma) \right\}\subset\Pset(\Gamma)$. Consider a trivial information structure $\tau_{0}$ in which both player's signal spaces are singletons. Note that $\tau_0$ has more-informed Buyer, but since we are interested in $\Pset$ rather than $\Pset^*$, we do not require Buyer's belief to be price independent. For any $(\pi_b,\pi_s) \in \Pset(\Gamma)$, define strategies and beliefs as follows. Let $p_l=\pi_s+\E[c(v)] \in [\ubar{v},\E[v]]$ and $p_h=\E[v]>\E[c(v)]$.
	\begin{itemize}
		\item Buyer's strategy:
	\begin{align*}
			\alpha(p_h)=\dfrac{p_l-\E[c(v)]}{p_h-\E[c(v)]}, 
			\alpha(p_l)=1, \text{ and } 
			\alpha(p)=\mathbf{1}_{p\le \underline{v}} \ \  \forall p \notin \{p_l,p_h\},
	\end{align*}
	where $\mathbf 1_{p\leq \underline{v}}$ denotes the indicator function for the set $\{p:p\leq \underline v\}$.
\item  Seller's strategy:
	\begin{align*}
		\sigma(p_l)=\frac{\pi_b}{\E[v]-p_l} \text{ and } \sigma(p_h)=1-\frac{\pi_b}{\E[v]-p_l}.
	\end{align*}
	Note that $\E[v]-p_l=S(\Gamma)-\pi_s\ge \pi_b$ guarantees that $\sigma(p_l),\sigma(p_h)\in[0,1]$.
\item  Beliefs:
	\begin{align*}
		\nu(v|p_l)=\nu(v|p_h)=\mu(v) \text{ and } \nu(v|p)= \delta_{\ubar{v}}(v) \ \ \forall p\notin\{p_l,p_h\}.
	\end{align*}
	\end{itemize}
	It is straightforward that the payoff from this strategy profile is $(\pi_b,\pi_s)$. So we need only verify that $(\sigma,\alpha,\nu)$ constitutes a wPBE. First, Buyer's strategy is optimal given beliefs because $\E[v]-p_h=0$, $\E[v]-p_l\ge 0$, and for any other price, Buyer's belief is a point mass on $\underline{v}$. Second, Seller's strategy is optimal: $\alpha(p_h)$ is defined such that $\alpha(p_h)(p_h-\E[c(v)])=p_l-\E[c(v)]$. So Seller is indifferent between offering $p_l$ and $p_h$. Any other price above $\ubar{v}$ is rejected and so is no better than $p_l$ and $p_h$. Seller's payoff is $p_l-\E[c(v)]=\pi_s\ge \piall(\Gamma)$, so any price below $\ubar{v}$ is also no better. Third, since Seller's strategy is type independent and $\nu=\mu$ on path, Bayes rule is satisfied on path.  

It remains to prove that $\Pset(\Gamma) \subset \left\{ (\pi_b,\pi_s):\pi_b\ge0,\pi_s\ge\piall(\Gamma),\pi_b+\pi_s\le S(\Gamma) \right\}$.
Pick any signal structure and any wPBE with belief $\nu$. Since $\Supp(\nu)\subset V$, sequential rationality implies that Buyer buys with probability one after any price $p<\ubar{v}$ and with probability zero when $p>\overbar{v}$. Therefore, Seller must obtain payoff $\pi_s \geq \underline \pi_s(\Gamma) \equiv \max\{\ubar{v}-\E[c(v)],0\}$. It is straightforward that Buyer's payoff $\pi_b\ge0$ and $\pi_b+\pi_s\le S(\Gamma)$.
\end{proof}

\section{Proof of \autoref{thm:joint:discrete}}
\label{appendix:thm:1:discrete}
We first prove \autoref{prop:joint:discrete} below, which we will use to prove \autoref{thm:joint:discrete}.

\subsection{A related result}
\begin{proposition}
	Fix any $\epsilon >0 $. $\exists$ a finite information structure such that $\forall(\pi_b,\pi_s)\in \{(\pi_b,\pi_s):\pi_b\ge\epsilon \text{, } \pi_s\ge \piall(\Gamma)+\epsilon \text{, and } \pi_b+\pi_s\le S(\Gamma)-\epsilon \}$, $\exists$ a finite price grid defining a game that has a sequential equilibrium with payoffs $(\pi_s,\pi_b)$.
		\label{prop:joint:discrete}
\end{proposition}

One aspect of this result is weaker than \autoref{thm:joint:discrete} because the price grid here varies with the equilibrium payoffs $(\pi_b,\pi_s)$. But another aspect is stronger: all payoffs $\epsilon$ away from the boundary of $\Pset(\Gamma)$ are obtained, rather than just an $\epsilon$-net of payoffs.

\begin{proof}
	Let us initially prove the statement assuming $\mathbb{P}(\ubar{v})>0$. First, choose any $\delta\in(0,\frac{1}{2})$ and any $\eta\in(0, \mathbb{P}(v=\ubar{v}))$. For now, we keep $\delta$ and $\eta$ as free parameters and we define the information structure and the corresponding equilibrium. At the end of the proof, we will verify that when $\delta$ and $\eta$ go to zero, the equilibrium payoffs span the target set of payoffs in \autoref{prop:joint:discrete}.
	
	We first define the {information structure}. Buyer is uninformed: $T_b=\left\{ \emptyset \right\}$. Seller gets two signals: $T_s=\left\{ l,h \right\}$, with distribution given by:
	\begin{align*}
		P(l,\emptyset,v)=&\eta\delta_{\ubar{v}}(v),\\
		P(h,\emptyset,v)=&\mu(v)-\eta\delta_{\ubar{v}}(v).
	\end{align*}
That is, $v=\ubar{v}$ is revealed to Seller with probability $\eta$ using signal ``$l$''. In the rest of the proof we omit $t_b$, as Buyer is uninformed. 

We next specify certain prices and a property of the finite price grid.
Choose any $\sigma_h\in[\delta,1-\delta]$. Define $$\overbar{p}_h=\frac{\E[v]-\eta\underline{v}}{1-\eta} \quad \text{and} \quad \ubar{p}_h=\frac{\eta\underline{v}+(1-\eta)\sigma_h\overbar{p}_h}{\eta+(1-\eta)\sigma_h}.$$ That is, $\overbar{p}_h$ is $\E[v|t_s=h]$ and $\ubar{p}_h$ is the expectation of $v$ conditional on the event that pools $\sigma_h$ proportion of $t_s=h$ with all $t_s=l$. Pick any $p_l\in\left[\max\left\{\frac{\E[c(v)]-\eta c(\underline{v})}{1-\eta},\underline{v}\right\},\frac{\eta\underline{v}+(1-\eta)\delta \overbar{p}_h}{\eta+(1-\eta)\delta}\right)$.\footnote{Fixing $\delta>0$, the interval is nonempty when $\eta$ is sufficiently close to $0$.} It holds that $\underline{v}\le p_l<\ubar{p}_h<\E[v]<\overbar{p}_h$. Consider any finite grid of prices that contains $\left\{ p_l,\ubar{p}_h,\overbar{p}_h \right\}$, and is otherwise arbitrary. 
	
Now we specify the strategy profile and beliefs, and verify equilibrium.	
	\begin{itemize}
		\item \emph{Case 1:} $c(\underline{v})>\E[c(v)]$. Seller's and Buyer's strategies, $\sigma$ and $\alpha$, and Buyer's beliefs $\nu$ are respectively:
\begingroup
\addtolength{\jot}{1em} 
\begin{align*}
&\sigma(\ubar{p}_h|h)=\sigma_h, \ \sigma(p_l|h)=1-\sigma_h, \text{ and }\sigma(\ubar{p}_h|l)=1;\\
&\alpha(\ubar{p}_h)=\frac{p_l-\frac{\E[c(v)]-\eta c(\underline{v})}{1-\eta}}{\ubar{p}_h-\frac{\E[c(v)]-\eta c(\underline{v})}{1-\eta}}, \alpha(p_l)=1, \text{ and } \forall p\notin\{p_l,\ubar p_h\}, \  \alpha(p)=\mathbf{1}_{p\le \underline{v}};\\
&\nu(v|\ubar{p}_h)=\frac{\eta(1-\sigma_h)\delta_{\underline{v}}(v)+\sigma_h \mu(v)}{\eta(1-\sigma_h)+\sigma_h}, \ \nu(v|p_l)=\frac{\mu(v)-\eta\delta_{\underline{v}}(v)}{1-\eta}, \text{ and } \forall p\notin\{p_l,\ubar p_h\}, \  \nu(v|p)=\delta_{\underline{v}}(v).
\end{align*}
\endgroup
That is, Seller with signal $h$ randomizes between prices $\ubar{p}_h$ and $p_l$, while after signal $l$ she chooses $\ubar{p}_h$. Buyer randomizes after price $\ubar p_h$, accepts $p_l$, and off-path accepts prices below $\ubar v$ and rejects otherwise.

Let us verify that $(\sigma,\alpha,\nu)$ is a sequential equilibrium. Buyer's sequential rationality is straightforward, as $\E_{\nu}[v|p_l]=\overbar{p}_h>p_l$, $\E_{\nu}[v|\ubar{p}_h]=\ubar{p}_h$ and $\E_{\nu}[v|p]=\underline{v}$ for any other $p$. For Seller, note that by definition of $\alpha(\ubar{p}_h)$, Seller with signal $h$ is indifferent between offering $p_l$ and $\ubar{p}_h$. Since $c(\underline{v})>\E[c(v)]$ by hypothesis, Seller with signal $l$ finds it strictly better offering $\ubar{p}_h$ than $p_l$. Any other price $p$ is worse than offering $p_l$ for both Seller types. Finally, for consistency of Buyer's belief: Bayes rule is straightforward on the equilibrium path. The off-path belief can be derived from the limit of Seller's fully mixed strategy $\widetilde{\sigma}_n(p|h)=\frac{n^2-1}{n^2}\sigma(p|h)+\frac{1}{n^2\times k}$ and $\widetilde{\sigma}_n(p|l)=\frac{n-1}{n}\sigma(p|l)+\frac{1}{n\times k}$, where $k$ is the number of prices in the grid.

\vspace{10pt}

		\item \emph{Case 2:} $c(\underline{v})\le\E[c(v)]$. Now consider:
\begingroup
\addtolength{\jot}{1em} 
\begin{align*}
&\sigma(\overbar{p}_h|h)=\sigma_h, \ \sigma(p_l|h)=1-\sigma_h, \text{ and }\sigma(p_l|l)=1;\\
&\alpha(\overbar{p}_h)=\frac{p_l-\frac{\E[c(v)]-\eta c(\underline{v})}{1-\eta}}{\overbar{p}_h-\frac{\E[c(v)]-\eta c(\underline{v})}{1-\eta}}, \ \alpha(p_l)=1, \text{ and } \forall p\notin\{p_l,\overbar p_h\}, \  \alpha(p)=\mathbf{1}_{p\le \underline{v}};\\
&\nu(v|\overbar{p}_h)=\frac{\mu(v)-\eta\delta_{\underline{v}}(v)}{1-\eta}, \ \nu(v|p_l)=\frac{\eta\sigma_h\delta_{\underline{v}}(v)+(1-\sigma_h)\mu(v)}{\eta\sigma_h+(1-\sigma_h)}, \text{ and } \forall p\notin\{p_l,\overbar p_h\}, \ \nu(v|p)=\delta_{\underline{v}}(v).
\end{align*}
\endgroup

That is, Seller with signal $h$ randomizes between prices $\overbar{p}_h$ and $p_l$, while after signal $l$ she chooses $p_l$. Buyer randomizes after price $\overbar p_h$, accepts $p_l$, and off-path accepts prices below $\ubar v$ and rejects otherwise.

Let us verify that $(\sigma,\alpha,\nu)$ is a sequential equilibrium. Buyer's sequential rationality is straightforward, as $\E_{\nu}[v|p_l]\ge \frac{\eta \underline{v}+(1-\eta)\delta \overbar{p}_h}{\eta+(1-\eta)\delta} >p_l$, \ $\E_{\nu}[v|\overbar{p}_h]=\overbar{p}_h$, and $\E_{\nu}[v|p]=\underline{v}$ for any other $p$. For Seller, note that by definition of $\alpha(\overbar{p}_h)$, Seller with signal $h$ is indifferent between offering $p_l$ and $\overbar{p}_h$. Since $c(\underline{v})\le\E[c(v)]$ by hypothesis, Seller with signal $l$ finds it weakly better offering $p_l$ than $\overbar{p}_h$. Any other price is worse than offering $p_l$ for both Seller types. Finally, for consistency of Buyer's belief: Bayes rule is straightforward on the equilibrium path. The off-path belief can be derived from the limit of Seller's fully mixed strategy $\widetilde{\sigma}_n(p|h)=\frac{n^2-1}{n^2}\sigma(p|h)+\frac{1}{n^2\times k}$ and $\widetilde{\sigma}_n(p|l)=\frac{n-1}{n}\sigma(p|l)+\frac{1}{n\times k}$, where $k$ is the number of prices in the grid.
	\end{itemize}
	Now we calculate the players' payoffs in the above equilibria. 
	\begin{itemize}
		\item \emph{Case 1:} $c(\underline{v})>\E[c(v)]$. In this case it is optimal for Seller to offer $p_l$ after signal $h$ and $\ubar{p}_h$ after signal $l$. Therefore, in equilibrium $$\pi_s=(1-\eta)p_l+\eta\ubar{p}_h-\E[c(v)]-\eta(1-\alpha(\ubar{p}_h))(\ubar{p}_h-c(\underline{v})).$$ Note that $\ubar{p}_h$ depends on $\sigma_h$ but $p_l$ does not. For any $\sigma_h\in[\delta,1-\delta]$, when $p_l=\max\left\{ \frac{\E[c(v)]-\eta c(\underline{v})}{1-\eta},\underline{v} \right\}$, $$\pi_s\le p_l -\E[c(v)]+\eta(\ubar{p}_h-p_l)\le\max\left\{ \frac{\eta}{1-\eta}(\E[c(v)]-c(\underline{v})),\underline{v}-\E[c(v)] \right\}+\eta(\E[v]-\underline{v}).$$ When $p_l=\frac{\eta\underline{v}+(1-\eta)\delta\overbar{p}_h}{\eta+(1-\eta)\delta}$, $$\pi_s\ge p_l-\E[c(v)]-\eta (\E[v]-\E[c(v)])= \frac{\eta\underline{v}+(1-\eta)\delta \overbar{p}_h}{\eta+(1-\eta)\delta}-\E[c(v)]-\eta S(\Gamma).$$
			Therefore, when $p_l$ traverses its domain, $\pi_s$ traverses a set containing the interval $$I_{s}=\Big[ \max\left\{ \frac{\eta}{1-\eta}(\E[c(v)-c(\underline{v})]),\underline{v}-\E[c(v)] \right\}+\eta(\E[v]-\underline{v}), \frac{\eta\underline{v}+(1-\eta)\delta \overbar{p}_h}{\eta+(1-\eta)\delta}-\E[c(v)]-\eta S(\Gamma) \Big).$$ In other words, $\forall\pi_s \in I_s$ and $\forall\sigma_h\in[\delta,1-\delta]$, there exists $p_l(\sigma_h)$ such that Seller's payoff is $\pi_s$. Now consider Buyer's payoff holding $\pi_s$ fixed. When $\sigma_h$ traverses $[\delta,1-\delta]$, $\pi_b$ changes continuously. If $\sigma_h=1-\delta$, then $\pi_b=(1-\eta)\delta(\overbar{p}_h-p_l)\le \delta (\overbar{v}-\underline{v})$. If $\sigma_h=\delta$, then with at most $\eta+\delta-\eta\delta$ probability the offer is rejected and hence $\pi_s+\pi_b\ge S(\Gamma)-(\eta+\delta)(\overbar{v}-\inf c(v))$. 
\vspace{10pt}			
		\item \emph{Case 2:} $c(\underline{v})\le\E[c(v)]$. In this case it is optimal for both Seller to offer $p_l$ no matter her signal, which induces Buyer to accept with probability $1$. Therefore, in equilibrium $\pi_s=p_l-\E[c(v)]$. Buyer is indifferent between accepting the offer or not at $\overbar{p}_h$. So Buyer gets positive payoff only when the price offered is $p_l$ and hence $\pi_b=(\eta\sigma_h \underline{v}+(1-\sigma_h)\E[v])-(\eta\sigma_h+(1-\sigma_h))p_l$.	Similar to Case 1, we can calculate that as $p_l$ traverses its domain, $\pi_s$ traverses the interval $$\Big[ \max\left\{ \frac{\eta}{1-\eta}(\E[c(v)-c(\underline{v})]),\underline{v}-\E[c(v)] \right\}, \frac{\eta\underline{v}+(1-\eta)\delta \overbar{p}_h}{\eta+(1-\eta)\delta}-\E[c(v)] \Big).$$ Holding any $p_l$ fixed, as $\sigma_h$ traverses the interval $[\delta,1-\delta]$, $\pi_b$ traverses $$\Big[ \eta(1-\delta)\underline{v}+\delta\E[v]-(\eta(1-\delta)+\delta)p_l,\eta\delta \underline{v}+(1-\delta)\E[v]-(\eta\delta+(1-\delta))p_l \Big].$$
	\end{itemize}
	It follows that in either case, as $\eta$ and $\delta$ converge to zero (with the order $\eta$ first and $\delta$ second), Seller's payoff that obtain across the family of equilibria we have constructed converges to $(\max\{ 0,\piall(\Gamma) \},S(\Gamma))$. For any such $\pi_s$, Buyer's payoff that obtain converges (uniformly) to $(0,S(\Gamma)-\pi_s)$. This completes the proof of \autoref{prop:joint:discrete} when $\mathbb{P}(\underline{v})>0$.
	
	When $\mathbb{P}(\underline{v})=0$, we first modify the original environment by pooling a small mass of valuations near $v=\underline{v}$ (which is feasible since $\underline v$ is the lowest value in the support of $V$). Call this modified environment $\widetilde{\Gamma}$. Plainly, $S(\widetilde{\Gamma})=S(\Gamma)$ and $\piall(\widetilde{\Gamma})\approx \piall(\Gamma)$. Therefore, $\forall \epsilon>0$, there exists such $\widetilde{\Gamma}$ such that $\Pset(\widetilde{\Gamma})$ covers all payoffs in $\Pset(\Gamma)$ that are more than $\frac{1}{2}\epsilon$ away from the boundary of $\Pset(\Gamma)$. We can now apply the previous argument with a positive probability of the lowest valuation and find an information structure $\widetilde{\tau}$ that implements all payoffs in $\Pset(\widetilde{\Gamma})$ that are more than $\frac{1}{2}\epsilon$ away from the boundary of $\Pset(\widetilde{\Gamma})$. The proof is completed by converting $\widetilde{\tau}$ to an information structure for the original environment $\Gamma$.
\end{proof}

\subsection{Proof of \autoref{thm:joint:discrete}}
\begin{proof}
	We utilize the construction in the proof of \autoref{prop:joint:discrete}. First, $\forall \varepsilon>0$, choose $\delta$ and $\eta$ as the corresponding parameters derived in \autoref{prop:joint:discrete} such that the implementable payoffs cover all points $\epsilon/2$ away from the boundary of $\Pset(\Gamma)$. Then, $\overbar{p}_h=\frac{\E[v]-\eta \underline{v}}{1-\eta}$. Choose grid size $\Delta\in\left( 0,\frac{1}{2}|\overbar{p}_h-\E[v]| \right)$. Construct an arbitrary grid of $[\underline{v},\bar{v}]$ with grid size $\Delta$. By the definition of grid size, there exists an on-grid price $\overbar{p}_h'\in\left[ \overbar{p}_h-\Delta ,\overbar{p}_h \right]$. Now choose $\eta'\le \eta$ s.t. $\overbar{p}_h'=\frac{\E[v]-\eta'\underline{v}}{1-\eta'}$. Note that $\overbar{p}_h'\ge \overbar{p}_h-\Delta$ implies $\frac{1-\eta'}{\eta'(\E[v]-\underline{v})}\le \frac{1}{\frac{\eta}{1-\eta}(\E[v]-\underline{v})-\Delta}$. From now on, we fix $\Delta,\eta',\delta,\bar{p}_h'$ and the grid.\par
	Pick any $(\pi_b,\pi_s)\in \Pset(\Gamma)$ that is $\varepsilon/2$ away from the boundary of $\Pset(\Gamma)$. Note that reducing $\eta$ to $\eta'$ expands the set of implementable payoffs in \autoref{prop:joint:discrete}. Therefore, given $\delta,\eta'$, the construction in \autoref{prop:joint:discrete} defines an information structure s.t. $(\pi_b,\pi_s)$ is an equilibrium payoff pair. Let $(p_l,\sigma_h)$ define the constructed equilibrium.\footnote{All other parameters defining the equilibrium are calculated from $\eta',p_l,\sigma_h$. Recall that the on-path prices are $p_l, \bar{p}_h',\underline{p}_h$.} Now we modify $p_l$ and $\sigma_h$ to ``snap'' the on-path prices onto the grid. Choose $p_l'$ to be the on-grid price no greater than and closest to $p_l$. So $p_l-p_l'<\Delta$. Let $\underline{p}_h'$ be the on-grid price closest to $\underline{p}_h$ such that $\sigma_h'=\frac{\eta'}{1-\eta'}\frac{\ubar{p}_h'-\ubar{v}}{\overbar{p}_h'-\ubar{p}_h'}\in[\delta,1-\delta]$ (note that since $\sigma_h'$ is increasing in $\underline{p}_h'$, this is achieved by one of the two grid points to the left and right of $\underline{p}_h$). It can be easily verified that $p_l'<\underline{p}_h'<\E[v]<\bar{p}_h'$. Observe that
	\begin{align*}\left|\frac{\d \sigma_h}{\d\ubar{p}_h'}\right|=\frac{\eta'}{(1-\eta')^2}\frac{\E[v]-\underline{v}}{(\overbar{p}_h'-\ubar{p}_h')^2}\le\frac{\eta'}{(1-\eta')^2}\frac{\E[v]-\underline{v}}{(\overbar{p}_h'-\E[v])^2}=  \frac{1}{\eta'(\E[v]-\underline{v})}\le \frac{1}{1-\eta'}\frac{1}{\frac{\eta}{1-\eta}(\E[v]-\underline{v})-\Delta},
	\end{align*}
	where the last inequality is from $\frac{1-\eta'}{\eta'(\E[v]-\underline{v})}\le \frac{1}{\frac{\eta}{1-\eta}(\E[v]-\underline{v})-\Delta}$. 
	Therefore, $|\underline{p}_h-\underline{p}_h'|\le \Delta$ implies $|\sigma_h-\sigma_h'|\le\frac{1}{1-\eta'}\frac{\Delta}{\frac{\eta}{1-\eta}(\E[v]-\underline{v})-\Delta}$. \par
	Take the information structure and equilibrium from the proof of \autoref{prop:joint:discrete} corresponding to parameters $\eta'$, $p_l'$ and $\sigma_h'$. Now we calculate the equilibrium payoffs and compare that to $(\pi_b,\pi_s)$. We discuss the two cases separately:
	
	In case 1 ($c(\underline{v})>\E[c(v)]$), we first bound $|\alpha'(\underline{p}_h')-\alpha(\underline{p}_h)|$:
	\begin{align*}
	    &|\alpha'(\underline{p}_h')-\alpha(\underline{p}_h)|\\
	    \le& \frac{|p_l'-p_l|}{\underline{p}_h-\E[c(v)|v>\underline{v}]}+\left|\frac{p_l'-\E[c(v)|v>\underline{v}]}{\underline{p}_h-\E[c(v)|v>\underline{v}]}-\frac{p_l'-\E[c(v)|v>\underline{v}]}{\underline{p}_h'-\E[c(v)|v>\underline{v}]}\right|\\
	    \le&\frac{2\Delta}{\underline{p}_h-\E[c(v)|v>\underline{v}]}.
	\end{align*}
	The second inequality follows from $\underline{p}_h'>p_l'$ and $|p_l'-p_l|,|\underline{p}_h'-\underline{p}_h|<\Delta$. In this case, Seller's payoff is $\alpha'(\underline{p}_h')(\underline{p}_h'-\E[c(v)])$ (note that Seller always finds $\underline{p}_h'$ optimal, which is accepted with probability $\alpha'$). Therefore,
	\begin{align*}
	    |\pi_s-\pi_s'|\le |\alpha'(\underline{p}_h')-\alpha(\underline{p}_h)|(\underline{p}_h-\E[c(v)])+\alpha'(\underline{p}_h')|\underline{p}_h'-\underline{p}_h|\le3\Delta,
	\end{align*}
	where the last inequality uses $\E[c(v)]>\E[c(v)|v>\underline{v}]$. Buyer's payoff is $\pi_b'=\E[v]-\underline{p}_h'+(1-\eta')(1-\sigma_h')(\underline{p}_h'-p_l')$ (note that Buyer always finds accepting the on-path prices optimal). Therefore,
	\begin{align*}
	    |\pi_b-\pi_b'|\le& |\underline{p}_h-\underline{p}_h'|+(1-\eta')\left(|\sigma_h'-\sigma_h|(\underline{p}_h-p_l)+(1-\sigma'_h)(|\underline{p}_h-\underline{p}_h'|+|p_l-p_l'|)\right)\\
	    \le&3\Delta+(\bar{v}-\underline{v})\frac{1}{1-\eta'}\frac{\Delta}{\frac{\eta}{1-\eta}(\E[v]-\underline{v})-\Delta}.
	\end{align*}
	
	In case 2 ($c(\underline{v})\le\E[c(v)]$), Seller finds it optimal to always offer $p_l$; hence, $\pi_s=p_l-\E[c(v)]$. Therefore,
	\begin{align*}
	    |\pi_s-\pi_s'|\le |p_l-p_l'|\le \Delta.
	\end{align*}
	Buyer's payoff is $\pi_b'=\E[v]-p_l'-(1-\eta')\sigma_h'(\bar{p}_h'-p_l')$ (note that Buyer always finds accepting the on-path prices optimal). Therefore,
	\begin{align*}
	    |\pi_b-\pi_b'|\le& |p_l-p_l'|+(1-\eta')\left(|\sigma_h'-\sigma_h|(\underline{p}_h-p_l)+\sigma'_h(|\underline{p}_h-\underline{p}_h'|+|p_l-p_l'|)\right)\\
	    \le&3\Delta+(\bar{v}-\underline{v})\frac{1}{1-\eta'}\frac{\Delta}{\frac{\eta}{1-\eta}(\E[v]-\underline{v})-\Delta}.
	\end{align*}	
	
	In either case, 
	\begin{align*}
	    ||(\pi_b,\pi_s)-(\pi_b',\pi_s')||\le 6\Delta+\frac{2\Delta(\bar{v}-\underline{v})}{\eta(\E[v]-\underline{v})-(1-\eta)\Delta}.
	\end{align*}
	By choosing $\Delta$ sufficiently small, we bound $||(\pi_b,\pi_s)-(\pi_b',\pi_s')||$ above by $\epsilon$. \par
	
	To summarize, $\forall \epsilon>0$, there exist parameters $\delta,\eta,\eta'$, and $\Delta$ such that for an arbitrary grid of $[\ubar{v},\bar{v}]$ with grid size $\Delta$, for any $(\pi_b,\pi_s)\in\Pset(\Gamma)$, we construct an information structure with sequential equilibrium payoff within the $\epsilon$-neighbouthood of $(\pi_b,\pi_s)$. That is, the set of payoffs from sequential equilibria corresponding to some information structure is an $\epsilon-$net of $\Pset(\Gamma)$.
\end{proof}

\section{Proof of \autoref{thm:seller}}
\label{sec:proof_seller}

We prove the theorem via three lemmas.

\begin{lemma}
		\label{lem:Seller:split}
		$\forall (\pi^*_b,\pi_s^*)\in \Pset^*_{us}(\Gamma)$, $\forall \pi_s\in[\pi_s^*,S(\Gamma)] $ and $\pi_b\in[0,S(\Gamma)-\pi_s]$, there exists $\widetilde{\tau}\in\mathbf{T}_{us}$ such that $(\pi_b,\pi_s)\in\Pi^*(\Gamma,\widetilde{\tau})$. 
\end{lemma}
In words, this lemma says that the set $\Pset^*_{us}(\Gamma)$ consists of all payoff pairs in $\Pset(\Gamma)$ such that Seller's payoff is above some floor. By definition, the floor is $\pius(\Gamma)$ defined in \eqref{e:ubar_pi}. \citet[Theorem 12, part 2]{hart2019better} implies that the floor is achieved. Hence, the lemma implies part \ref{seller2} of \autoref{thm:seller}. 

\begin{proof}
	Let $\gamma=S(\Gamma)-(\pi_b+\pi_s)$ be the loss of total surplus for payoff pair $(\pi_b,\pi_s)$. We construct an information structure $\widetilde{\tau}$ such that the efficiency loss is $\gamma$, Seller's payoff is $\pi_s$ and Buyer's payoff is $\pi_b$. Let $P(t_b,v)$ be the joint distribution specified by an Seller-uninformed information structure $\tau$ for which $(\pi_b^*,\pi_s^*)\in\Pi^*(\Gamma,\tau)$.
	
	First we determine the types that are not traded. For this, we find a threshold value $z^*$ such that trading all expected valuations strictly below $z^*$ and some fraction of expected valuation $z^*$ generates surplus $\gamma$. Consider the function
	\begin{align*}
		y(z)=\int_{\E[v|t_b]<z}(v-c(v))P(\d t_b,\d v),
	\end{align*}
which is well defined because the domain of integration is measurable. The set $\left\{ t_b|\E[v|t_b]<z \right\}$ expands when $z$ increases. So $y(z)$ is increasing in $z$. Moreover, $y(\infty)=S(\Gamma)$ and $y(-\infty)=0$. So there exists $z^*$ such that $y(z)\leq (\geq) \gamma$ for $z<(>)z^*$. 
By definition, $\left\{ t_b|\E[v|t_b]<z \right\}=\bigcup_{\epsilon>0}\left\{ t_b|\E[v|t_b]<z-\epsilon \right\}$, so $y(z)$ is a left-continuous function. Hence, $y(z^*)\le \gamma$. Define $\beta$ by:
	\begin{align*}
		\gamma=y(z^*)+\beta\int_{\E[v|t_b]=z^*}(v-c(v))P(\d t_b,\d v)
	\end{align*}
	The RHS is $y(z^*)\le \gamma$ when $\beta=0$ and $\lim_{z\to z^{*+}}y(z)\ge \gamma$ when $\beta=1$. So $\beta\in[0,1]$. In words, excluding all $t_b$ that induces $\E[v|t_b]<z^*$ and $\beta$ portion of $t_b$ inducing $\E[v|t_b]=z^*$ leads to efficiency loss $\gamma$.
	
Next, we construct $\widetilde{\tau}\in \mathbf T_{us}$ such that all remaining surplus is realized and Seller gets payoff $\pi_s$. If Seller sells at price $p$ and trades with all remaining types, Seller's payoff is:
	\begin{align*}
		\int_{\E[v|t_b]>z^*}(p-c(v))P(\d t_b,\d v)+(1-\beta)\int_{\E[v|t_b]=z^*}(p-c(v))P(\d t_b,\d v).
	\end{align*}
	Therefore, Seller's payoff is $\pi_s$ when trading with all remaining types at the price\footnote{Note that $\pi_s\le S(\Gamma)-\gamma$ guarantees that $p^*\ge z^*$.}
	\begin{align*}
		p^*=\dfrac{\pi_s+\int_{\E[v|t_b]>z^*}c(v)P(\d t_b,\d v)+(1-\beta)\int_{\E[v|t_b]=z^*}c(v)P(\d t_b,\d v)}{		\int_{\E[v|t_b]>z^*}P(\d t_b,\d v)+(1-\beta)\int_{\E[v|t_b]=z^*}P(\d t_b,\d v)}.
	\end{align*}
	To ensure that all non-excluded Buyer types accept price $p^*$, we construct $\widetilde{\tau}$ by pooling all non-excluded types such that $\E[v|t_b]<p^*$ and a $\lambda$ fraction of those with signal $t_b$ such that $\E[v|t_b]\ge p^*$. The fraction $\lambda $ is determined as follows:
	\begin{align*}
		&\lambda\int_{\E[v|t_b]\ge p^*}(v-p^*)P(\d t_b,\d v)+\int_{z^*<\E[v|t_b]<p^*}(v-p^*)P(\d t_b,\d v)\\
		&+(1-\beta)\int_{\E[v|t_b]=z^*}(v-p^*)P(\d t_b,\d v)=0\\
		\implies&\lambda=\dfrac{\int_{z^*<\E[v|t_b]<p^*}(p^*-v)P(\d t_b,\d v)+(1-\beta)\int_{\E[v|t_b]=z^*}(p^*-z^*)P(\d t_b,\d v)}{\int_{\E[v|t_b]>p^*}(v-p^*)P(\d t_b,\d v)},
	\end{align*}
	where $\lambda\in[0,1]$ follows from the fact that the LHS of the first equality traverses from negative to positive when $\lambda$ traverses $[0,1]$.
	
Let $\widetilde{T}_b=T_b\cup\{t_{\emptyset}\}$, where $t_{\emptyset}$ is topologically disjoint from $T_b$. The information structure $\widetilde \tau\in \mathbf T_{us}$ is given by the following distribution:
	\begin{align*}
			\widetilde{P}(t_b,v)&=
			\begin{dcases}
			(1-\lambda)P(t_b,v)&\forall\ t_b\ s.t.\ \E[v|t_b]\ge p^*\\
			P(t_b,v)&\forall\ t_b\ s.t.\  \E[v|t_b]<z^*\\
			\beta P(t_b,v)&\forall\ t_b\ s.t.\  \E[v|t_b]=z^*;\\
			\end{dcases}\\[10pt]
		\widetilde{P}(t=t_{\emptyset},v)&=\lambda\int_{\E[v|t_b]\ge p^*}P(\d t_b, v)+\int_{z^*<\E[v|t_b]<p^*}P(\d t_b,v)+(1-\beta)\int_{\E[v|t_b]=z^*}P(\d t_b,v).
	\end{align*}
(It can be verified that $\widetilde{P}$ defines a valid information structure.)
 
Now we define Buyer's strategy $\widetilde \alpha$. Let $\alpha$ be Buyer's strategy corresponding to the wPBE of game $(\Gamma,\tau)$. Define $\widetilde{\alpha}(p,t_b)=\alpha(p,t_b)$ when $t_b\neq t_{\emptyset}$, and $\widetilde{\alpha}(p,t_{\emptyset})=\mathbf{1}_{p\le p^*}$. Sequential rationality of $\widetilde \alpha$ is straightforward. 

It remains only to verify that pricing at $p^*$ is optimal for Seller. There is no profitable deviation to any higher price because
$$\sup_{p>p^*}\int \widetilde{\alpha}(p,t_b)(p-c(v))\widetilde{P}(\d t_b,\d v)=(1-\lambda)\int \alpha(p,t_b)(p-c(v)){P}(\d t_b,\d v) \le (1-\lambda)\pi_s^*.$$
There is no profitable deviation to any price lower than $z^*$ because
$$\sup_{p\le z^*}\int \widetilde{\alpha}(p,t_b)(p-c(v))\widetilde{P}(\d t_b,\d v)
		\le \sup_{p\le z^*}\int \alpha(p,t_b)(p-c(v)){P}(\d t_b,\d v)
		\le \pi_s^*.$$
	By construction, there is no $t_b$ that induces a belief with $\E_{\widetilde{\nu}}[v|t_b]\in(z^*,p^*)$. Therefore, it is suboptimal for Seller to post any price in $(z^*,p^*)$. It follows that it is optimal (strictly optimal when $\pi_s>\pi_s^*$) for Seller to offer $p^*$ and get payoff $\pi_s$. Buyer's payoff is all the remaining surplus:	 $S(\Gamma)-\gamma-\pi_s=\pi_b$.
	\begin{remark}
		If $\pi_b=S(\Gamma)-\pi_s$, the market is efficient, $z^*=-\infty$, and $p^*=\inf_{t_b}\E[v|t_b]$. \qedhere
	\end{remark}
\end{proof}

\begin{lemma}
\label{lem:us_triangle}
$\Pset_{us}^*(\Gamma)=\Pset^*(\Gamma)$. 	
\end{lemma}
In words, this lemma says that uninformed-Seller information structures implement all payoff pairs implementable with price-independent beliefs under any information structure. As it is trivial that $\Pset^*_{us}(\Gamma)\subset \Pset^*_{mb}(\Gamma)\subset \Pset^*(\Gamma)$, this establishes part \ref{seller1} of \autoref{thm:seller}.
\begin{proof}
	$\Pset_{us}^*(\Gamma)\subset \Pset^*(\Gamma)$ is trivial, so we need only prove the opposite direction. Suppose that under some information structure $\tau$ there is a wPBE $(\sigma,\alpha,\nu)$ with price-independent beliefs and payoffs $(\pi_b,\pi_s)$. Consider an information structure $\tau'\in\mathbf{T}_{us}$ defined by $Q(t_b,v)=P(T_s,t_b,v)$. $\nu$ is a consistent belief system given information structure $\tau$ and strategy $\sigma$. Now we verify that $\forall \sigma'$, $\nu$ is a consistent belief system given $\tau'$ and $\sigma'$. For every measurable rectangle $D\subset \mathbb{R}\times T_b\times V$,
\begin{align*}
	 &\int_D\nu(\d v|p,t_b)\sigma'(\d p) Q(\d t_b, V)\\
	=&\int_{D_p}\sigma'(\d p)\cdot\int_{D_{t_b,v}}\nu(\d v|p,t_b) P(T_s,\d t_b, V)\\
	=&\int_{D_p}\sigma'(\d p)\cdot \int_{D_{t_b,v}}P(T_s,\d t_b,\d v)\\
	=&\int_D\sigma'(\d p) Q(\d t_b,\d v),
\end{align*}
where $D_{p}$ and $D_{t_b,v}$ are the projection of $D$ on dimension $p$ and $t_b,v$ respectively. The first and third equalities use the definition of measure $Q$. The second equality is the definition of price-independent belief. Since the product-sigma-algebra is uniquely defined by the product of sigma-algebras, verifying on all rectangular $D$ guarantees that $\nu$ is a consistent belief system. Therefore, $\alpha$ remains a best response for Buyer. Moreover,
	\begin{align*}
		&\sup_{\sigma'}\int (p-c(v))\alpha (p,t_b)\sigma'(\d p)Q(\d t_b,\d v)\\
		=&\sup_{\sigma'}\int (p-c(v))\alpha (p,t_b)\sigma'(\d p)P(\d t_s,\d t_b,\d v)\\
		\le&\int (p-c(v))\alpha (p,t_b)\sigma(\d p|t_s)P(\d t_s,\d t_b,\d v)=\pi_s.
	\end{align*}
	The first line is achievable by a Seller's strategy when $\alpha$ is modified to break ties in favor of Seller. Therefore, $\pi_s\ge \pius(\Gamma)$ and hence $(\pi_b,\pi_s)\in \Pset^*_{us}(\Gamma)$.
\end{proof}
	
\begin{lemma}
\label{lem:us_unique} For any $(\pi_b,\pi_s)\in \Pset^*_{us}(\Gamma)$ with $\pi_s>\pius(\Gamma)$, there is $\tau\in \mathbf{T}_{us}$ with $\Pi(\Gamma,\tau)=\left\{ (\pi_b,\pi_s) \right\}$.
\end{lemma}
This ``unique implementation'' lemma corresponds to part \ref{seller3} of \autoref{thm:seller}.
\begin{proof}
	We have established that $\pius(\Gamma)$ is achieved in an equilibrium. Use $\pius(\Gamma)$ as the $\pi_s^*$ in the proof of \autoref{lem:Seller:split} and construct the corresponding information structure. Note that given the information structure, Seller's payoff from any deviation to a price other than $p^*$ is bounded above by $\pi_s^*<\pi_s$. As a result $p^*$ is the uniquely optimal price given Buyer's best response $\widetilde{\alpha}$.\par
	
Now we show that for any other $\alpha'$ that is sequentially rational, Seller's payoff is still bounded above by $\pi^*_s$. Suppose not, to contradiction. Then there is $p$ such that $\int \alpha'(p,t_b)(p-c(v))P(\d t_b,\d v)>\pi^*_s$. Let $T_p$ be the subset of all Buyer's signals $t_b$ for which $\E[v|t_b]=p$ --- signals making Buyer indifferent between buying or not. Note that any two sequentially rational Buyer strategies differ only on $T_p$. We have:
 \begin{align*}
	 &\lim_{p'\to p^-}\int\tilde\alpha(p',t_b)(p'-c(v))P(\d t_b,\d v)\\
 &=\int_{\E[v|t_b]\ge p} (p-c(v))P(\d t_b,\d v)\\
 &=\int\alpha'(p,t_b)(p-c(v)) P(\d t_b,\d v )+\int_{t_b\in T_p}(1-\alpha'(p,t_b))(p-c(v)) P(\d t_b,\d v )\\
 &\ge \int\alpha'(p,t_b)(p-c(v)) P(\d t_b,\d v )>\pi_s^*.
 \end{align*}
 The first two equalities are from the fact that $\tilde\alpha$ and $\alpha'$ differ from $\mathbf{1}_{\E[v|t_b]\ge p}$ only on $T_p$. The inequality is from $\E[v|t_b]=p$ on $T_p$, $v\ge c(v)$ and $\alpha'\le 1$. This implies that there exists $p'<p$ giving Seller payoff strictly above $\pi^*_s$, which is a contradiction.

Therefore, when $\pi_s>\pi_s^*$, the information structure constructed in \autoref{lem:Seller:split} implements the unique equilibrium payoff pair $(\pi_b,\pi_s)$. The result follows from choosing $\pi_s^*=\pius(\Gamma)$.
\end{proof}

\section{Proof of \autoref{thm:buyer}}\label{sec:proof:fb}
\begin{proof}
	We first introduce some notations. Given the continuous cost function $c(v)$, any Buyer belief $\nu\in \Delta(V)$, and any price $p\in V$, define
	\begin{align*}
		\pi(c,\nu,p)&=\int_{v\ge p}(p-c(v))\nu(\d v),\\
		\pi^*(c,\nu)&=\max_{p\in V}\pi(c,\nu,p),\\
		\sigma^*(c,\nu)&=\argmax_{p\in V}\pi(c,\nu,p).
	\end{align*}
	In words, $\pi$ is Seller's payoff from offering an arbitrary price $p$, $\pi^*$ is Seller's payoff from an optimal price, and $\sigma^*$ is the set of optimal prices. Our assumption that $v-c(v)\ge0$ implies that $\pi(c,\nu,p)$ is a left-continuous function of $p$ that only jumps up, and hence it is upper semi-continuous. Therefore, $\pi^*$ is well-defined and $\sigma^{*}(c,\nu)$ is nonempty and compact.
	
	We prove \autoref{thm:buyer} in 4 steps. In step 1, we define a discretized environment for a grid size $d$. In step 2, we construct a distribution of IPDs for the discretized environment. In step 3, we show that as $d \to 0$, the distributions in step 2 converge to a distribution of IPDs whose expectation is the prior $\mu$. In step 4, we construct an information structure and equilibrium for the original environment utilizing the distribution derived in step 3.
	
	\emph{Step 1}. We discretize the problem. Pick any $d>0$. Discretize the support $V$ to a grid $V'=\left\{ v_1,\dots,v_n \right\}$ such that $v_{i+1}-v_i<d$, $v_1\le \underline{v}$ and $v_n>\overbar{v}$.  Let $p^*$ be an element of $\sigma^*(c,\mu)$ and include $p^*$ in $V'$. Define
		\begin{align*}
			\mu'_i&=\int\mathbf{1}_{v_i\le v<v_{i+1}}\mu(\d v),\\
			c'(v_i)&=\int\mathbf{1}_{v_i\le v<v_{i+1}}c(v)\mu(\d v).
	\end{align*}
	Now consider a new environment $\Gamma'=(c',\mu')$ with the discrete support $V'$. $\Gamma'$ augments $\Gamma$ by grouping all Buyer types in interval $[v_i,v_{i+1})$ and assuming Buyer behaves as if the valuation is $v_i$. A key property of the environment $\Gamma'$ is that  $\forall v_i\in V'$, $\pi(\Gamma',v_i)=\pi(\Gamma,v_i)$, that is, Seller's payoff from offering on-grid prices is invariant under the environment change. Since $p^*\in V'$, $\pi^*(\Gamma')=\pi^*(\Gamma)$.\par
	\emph{Step 2}. The following lemma---whose proof is provided after the current proof is completed---implies that that there exist IPDs $\left\{ \nu_j \right\}_{j=1}^J$ and $\left\{ p_j \right\}\in \Delta(J)$ such that $\sum p_j\nu_j=\mu'$ and $\sigma^*(c',\mu')\subset \cap\sigma^*(c',\nu^j)$.

\begin{lemma}\label{lem:idIPD}
 When $\Supp(\mu)$ is finite, there exists IPDs $\{\nu_j\}_{j=1}^J$ and $\{q_j\}\in \Delta(J)$ such that $\sum q_j \nu_j=\mu$ and $\sigma^*(c,\mu)\subset \sigma^*(c,\nu_j)$.
\end{lemma}

	\emph{Step 3}. For each $d_n=\frac{1}{2n}$, go through Steps 1--2 and construct a collection $\left\{ p^j,\nu^j \right\}$. This collection  resembles a probability measure $P_n\in \Delta^2(V)$. By construction, any $\nu\in \Supp(P_n)$ is an IPD satisfying $p^*\in \sigma^*(c,\nu)$. 	We use the following lemma---whose proof is provided after the current proof is completed---to construct a measure $P^*$ whose support contains only those IPDs such that $p^*$ is an optimal price ( $P^*$ is a limit point of $P_n$ under the weak topology.
 
\begin{lemma}\label{lem:IPD:weak:converge}
    Suppose the sequence $(P_n)$ in $\Delta^2(V)$ satisfies $\int \nu P_n(\d \nu)\xrightarrow[]{w} \mu$ and $\forall \nu\in\Supp(P_n)$, $\nu$ is an IPD satisfying $p^*\in \sigma^*(c,\nu)$. Then, $\exists P^*\in \Delta^2(V)$ s.t. $\int \nu P^*(\d \nu)=\mu$ and $\forall \nu\in\Supp(P^*)$, $\nu$ is an IPD satisfying $p^*\in \sigma^*(c,\nu)$.
\end{lemma}
	Then,
	\begin{align}
		\int \pi^*(c,\nu)P^*(\d \nu)=\int \pi(c,\nu,p^*)P^*(\d \nu)=\pi^*(c,\mu).
		\label{eqn:thm:3:1}
	\end{align}

	\emph{Step 4}. Now we define a fully-informed Buyer information structure that implements any $(\pi_b,\pifb(\Gamma))\in \Pset(\Gamma)$. Let $\beta\in[0,1]$ satisfy $\pi_b=\beta (S(\Gamma)-\pifb(\Gamma))$. Take the signal space $T_s=\Delta(V)$ and define the signal distribution by $\int_D P(\d t_s,\d v)=\int_D t_s(\d v)P^*(\d t_s)$. That is, the information structure $\tau=(T_s,P) \in \mathbf{T}_{fb}$ induces Seller's belief $\nu$ according to distribution $P^*(\nu)$.
	
	Buyer's strategy is $\alpha(v,p)=\mathbf{1}_{v\ge p}$, which is obviously optimal. Seller's strategy is $\sigma(p|t_s=\nu)=\beta \delta_{p=\min\Supp(\nu)}+(1-\beta)\delta_{p=\max\Supp(\nu)}$. Then $\forall \sigma'$:
	\begin{align*}
		\int (p-c(v))\mathbf{1}_{v\ge p}\sigma'(\d p|t_s)P(\d t_s,\d v)=&\int (p-c(v))\mathbf{1}_{v\ge p}\sigma'(\d p|\nu)\nu(\d v)P^*(\d \nu)\\
		=&\int \pi(c,\nu,p)\sigma'(\d p|\nu) P^*(\d \nu)\\
		\le&\int \pi^*(c,\nu)P^*(\d \nu)  \qquad \text{\small{because $\pi(c,\nu,p)\le \pi^*(c,\nu)$}},
	\end{align*}
	where the equalities are accounting identities. Meanwhile, Seller's payoff using strategy $\sigma$ is
	\begin{align*}
		\int (p-c(v))\mathbf{1}_{v\ge p}\sigma(\d p|t_s)P(\d t_s,\d v)=&\int \beta\pi(c,\nu,\min\Supp(\nu))+(1-\beta)\pi(c,\nu,\max\Supp(\nu)) P^*(\d \nu)\\
		=&\int \pi^*(c,\nu)P^*(\d \nu)=\pifb(\Gamma),
	\end{align*}
	where the second equality is from $\nu$ being an IPD and the third equality is from \autoref{eqn:thm:3:1}. Therefore, $\sigma$ is optimal for Seller and Seller's equilibrium payoff is $\pifb$. Buyer's payoff is:
	\begin{align*}
		&\int (v-p)\mathbf{1}_{v\ge p}\sigma(\d p|t_s)P(\d t_s,\d v)\\
		=&\int\beta (v-\min\Supp(\nu))\mathbf{1}_{v\ge \min\Supp(\nu)}+(1-\beta)(v-\max\Supp(\nu))\mathbf{1}_{v=\max\Supp(\nu)}\nu(\d v)P^*(\d \nu)\\
		=&\int \beta (v-c(v)-(\min\Supp(\nu)-c(v)))\nu(\d v)P^*(\d \nu)\\
		=&\beta(S(\Gamma)-\int \pi^*(c,\nu)P^*(\d \nu))=\pi_b,
	\end{align*}
	where second equality is from $v-\max\Supp(\nu)\mathbf{1}_{v=\max\Supp(\nu)}=0$, the third equality is from $\nu\in\Supp(P^*)$ being an IPD, and the last equality is from \autoref{eqn:thm:3:1}.
	
	To sum up, we construct $\tau\in \mathbf{T}_{fb}$ such that $(\pi_b,\pifb(\Gamma))\in \Pi^*(\Gamma,\tau)$. Since $(0,S(\Gamma))$ can be implemented by perfect revealing $v$ to Seller, any other $(\pi_b,\pi_s)$ in $\Pset^*_{fb}(\Gamma)$ can be implemented by public randomization, which means Buyer is better informed in the strong sense that his information is a refinement of Seller's.
\end{proof}

\begin{proof}[Proof of \autoref{lem:idIPD}]
    We prove the result by induction. When $|\Supp(\mu)|=1$, the statement is trivially true. Now we assume by induction that the statement is true for $|\Supp(\mu)|\le n$ and prove it for $|\Supp(\mu)|=n+1$. Let $V=\Supp(\mu)=\{v_1,\ldots,v_{n+1}\}$. We discuss two cases separately: \par
  \begin{itemize}  
    \item\emph{Case 1}: $v_i>c_i$ for all $i\le n$. Define $\hat{\nu}_{n+1}=1$ and recursively define
    $$\hat{\nu}_i=\frac{\sum_{j=i+1}^{n+1}\hat{\nu}_j\cdot(v_{i+1}-v_i)}{v_i-c_i}$$
    for $i=n\ldots 1$. Normalize $\{\hat{\nu}_i\}$ to a probability vector $\nu=\frac{1}{\sum_i \hat{\nu}_i}\hat{\nu}$. Then, it is easy to verify that $\nu\in\Delta V$ and $\nu$ is an IPD:
    \begin{align*}
        \pi(c,\nu,v_{i+1})-\pi(c,\nu,v_{i})=\sum_{j=i+1}^{n+1}\nu_j\cdot(v_{i+1}-v_i)-\nu_i\cdot(v_i-c_i)=0,\ \forall i.
    \end{align*}
     Therefore, $\sigma^*(c,\nu)\supset V$.

    \item\emph{Case 2}: $v_{i}=c_{i}$ for some $i\le n$. Let $i_0$ be the smallest $i$ such that this is true. Define $\hat{\nu}_{i_0}=1$ and recursively define $\hat{\nu}_i=\frac{\sum_{j=i+1}^{n+1}\hat{\nu}_j\cdot(v_{i+1}-v_i)}{v_i-c_i}$ for $i=1,\ldots, i_0$. Normalize $\{\hat{\nu}_i\}$ to $\nu=\frac{1}{\sum_i \hat{\nu}_i}\hat{\nu}$. Then, the exactly same argument as in Case 1 implies that $\nu\in\Delta V$ and $\nu$ is an IPD. Moreover, since $v_{i_0}=c_{i_0}$, $\pi^*(c,\nu)=0$. Therefore, $\sigma^*(c,\nu)\supset \Supp(\nu)\cup [v_{i_0},+\infty)\supset V$.
    \end{itemize}
     Next, we ``remove $\nu$ from $\mu$'' to reduce its support size. Let $q=\min \left\{ \frac{\mu_i}{\nu_i}\right\}$ and $\hat{\mu}=\mu-q\cdot \nu$. By definition, $|\Supp(\hat{\mu})|\le n$. Normalize $\hat{\mu}$ to $\mu'=\frac{1}{\sum_i\hat{\mu}_i}\hat{\mu}$. Then, $\forall i,i'\in(1,\ldots,n+1)$,
    \begin{align*}
        \left(\sum_{\ell}\hat{\mu}_{\ell}\right)\left(\pi(c,\mu',v_i)-\pi(c,\mu',v_{i'})\right)=&\pi(c,\mu,v_i)-\pi(c,\mu,v_{i'})-q\cdot \left(\pi(c,\nu,v_i)-\pi(c,\nu,v_{i'})\right)\\
        =&\pi(c,\mu,v_i)-\pi(c,\mu,v_{i'})\\
        \implies \sigma^*(c,\mu')=\sigma^*(c,{\mu}).&
    \end{align*}
    The first equality is from the linearity of $\pi$ and the second equality is from $\sigma^*(c,\nu)\supset V$. Then, by induction, there exists IPDs $\nu_j$ and $q_j$ s.t. $\sum q_j\nu_j=\mu'$ and $\sigma^*(c,\mu)=\sigma^*(c,\mu')\subset\sigma^*(c,\nu_j)$. Therefore, the statement is proved by appending $\nu$ to $(\nu_j)$, and normalizing the probability to $\left(q_j\cdot \sum\hat{\mu_i},q \right)$.
\end{proof}

\begin{proof}[Proof of \autoref{lem:IPD:weak:converge}]
    
	By Prokhorov's theorem, there exists a convergent subsequence of $P_n$; without loss we suppose $P_n\xrightarrow{w}P^*$ (i.e., weak convergence, which is implied by convergence in the Prokhorov metric). Let $\mu_n=\int \nu P_n(\d \nu)$. By assumption, $\mu_n\xrightarrow[]{w}\mu$. It follows that $\int \nu P^*(\d \nu)=\mu$.\footnote{\label{fn:seller:bayes}For any continuous $h(v)$, $\nu\mapsto \int h(v)\nu(\d v)$ is a bounded an continuous function on $\Delta(V)$ under the Prokhorov metric. Therefore, since $\mu_n \xrightarrow{w}\mu$ and $P_n\xrightarrow{w}P^*$,
	\begin{align*}
		\int h(v)\mu_n(\d v)\to \int h(v)\mu(\d v) \quad \text{ and } \quad \iint h(v)\nu(\d v)P_n(\d \nu)\rightarrow\iint h(v)\nu(\d v) P^*(\d \nu).
	\end{align*}
Since $\int h(v)\mu_n(\d v) = \iint h(v)\nu(\d v)P_n(\d \nu)$, it follows that
$\int \nu P^*(\d \nu)=\mu$.
}

Now we show that $\forall \nu\in\Supp(P^*)$, $\nu$ is an IPD. First, \autoref{lem:convergence:support} shows that there exists a sub-sequence $n_k\to\infty$ and  $\nu_{n_k}\in\Supp(P_{n_k})$ such that $\nu_{n_k}\xrightarrow{w}\nu$. Then, $\forall p\in \Supp(\nu)$, there exists a sub-sequence $n_{k_{s}}\to \infty$ and $p_{n_{k_{s}}}\in\Supp(\nu_{n_{k_{s}}})$ such that $p_{n_{k_{s}}}\to p$. \autoref{lem:converge:1} proves that $\pi(c,\nu,p)\ge \varlimsup\pi(c,\nu_{n_{k_{s}}},p_{n_{k_{s}}})=\varlimsup\pi^*(c,\nu_{n_{k_{s}}})$. The equality is from $\nu_{n_{k_{s}}}$ being IPD and $p_{n_{k_s}}$ being in its support. \citet[Theorem 12, part 1]{hart2019better} proves that $\pi^*(c,\nu)\le\varliminf \pi^*(c,\nu_{n_{k_{s}}})$. Therefore, $\pi(c,\nu,p)=\pi^*(c,\nu)$ and hence $\nu$ is an IPD. 
	
	In the previous analysis, if we pick $p=p^*$, then since $p^*\in \Supp(\nu_{n_k})$, it follows that $p_{n_k}= p^*$ and hence trivially $p_{n_k}\to p^*$. Therefore $\pi(c,\nu,p^*)=\pi^*(c,\nu)$.
\end{proof}

\begin{lemma}
	\label{lem:convergence:support}
	Let $(S,\rho)$ be a separable metric space, $\left\{ P_n \right\}\subset \Delta(S)$ and $P_n\xrightarrow{w}P$. Then $\forall s\in\Supp(P)$, $\exists$ sequence $s_{n_k}\in \Supp(P_{n_k})$ s.t. $n_k\to \infty$ and $s_{n_k}\xrightarrow{\rho}s$. 
\end{lemma}
\begin{proof}
	For any $s\in \Supp(P)$, suppose towards contradiction that the statement is not true. Then we claim that $\exists \epsilon>0$, $N\in\mathbb{N}$ s.t. $\forall n\ge N$ $\Supp(P_n)\bigcap B_{\epsilon}(s)=\emptyset$. Otherwise, $\forall \epsilon>0$, $N\in\mathbb{N}$ exists $n\ge N$ s.t. $\Supp(P_n)\bigcap B_{\epsilon}(s)\neq\emptyset$ $\implies$ pick any $N=k$ and $\epsilon=\frac{1}{k}$, there exists $n_k\ge k$ and $s_{n_k}\in \Supp(P_{n_k})$ s.t. $\rho(s,s_{n_k})<\frac{1}{k}$ and hence the assumption is not true.
	
	Since $\exists \epsilon>0$ and $N$ s.t. $\forall n\ge N$ $\Supp(P_n)\bigcap B_{\epsilon}(s)=\emptyset$, this implies $\varliminf P_n(B_{\epsilon}(s))=0\ge P(B_{\epsilon}(s))$ (by the Portmanteau theorem). This contradicts the assumption that $s\in \Supp(P)$.
\end{proof}

\begin{lemma}
	\label{lem:converge:1}
	Let $c\in C(V)$, $\left\{ \nu_n \right\}\subset \Delta(V)$, $\left\{ p_n \right\}\subset V$. If $\nu_n\xrightarrow{w}\nu$ and $p_n\to p$, then
	\begin{align*}
		\pi(c,\nu,p)\ge \varlimsup \pi(c,\nu_n,p_n).
	\end{align*}
\end{lemma}
\begin{proof}
	Define
	\begin{align}
		h_{\delta,p}(v)=\frac{v-p+\delta}{\delta}\wedge [0,1],\label{eqn:h1}
	\end{align}
	where $\cdot\wedge[0,1]$ is the truncation functional on $[0,1]$. Then $h_{\delta,p}$ is a continuous and bounded function and $\mathbf{1}_{v\ge p}\le h_{\delta,p}(v)\le \mathbf{1}_{v\ge p-\delta}$.
Then $\forall \eta>\delta>0$:
	\begin{align*}
		\int_{p-\eta}^{\infty}(p-c(v))\nu(\d v)\ge&\int h_{\delta,p-\eta+\delta}(v)(p-c(v)\nu(\d v))\\
		=&\lim_{n\to \infty} \int h_{\delta,p-\eta+\delta}(v)(p-c(v))\nu_n(\d v)\\
		\ge&\varlimsup_{n\to \infty} \int_{p-\eta+\delta}^{\infty}(p-c(v))\nu_n(\d v)\\
		\ge&\varlimsup_{n\to \infty} \int_{p-\eta+\delta}^{\infty}(p+\eta-c(v))\nu_n(\d v)-\eta\\
		\ge&\varlimsup_{n\to \infty} \int_{p_n}^{\infty}(p+\eta-c(v))\nu_n(\d v)-\eta\\
		\ge&\varlimsup_{n\to \infty} \int_{p_n}^{\infty}(p_n-c(v))\nu_n(\d v)-\eta.
	\end{align*}
	The first inequality above is because $h_{\delta,p-\eta+\delta}(v)\le \mathbf{1}_{v\ge p-\eta}$ and $\forall v\in [p-\eta,p-\eta+\delta]$, $p>v\ge c(v)$. The first equality is from $\nu^n\xrightarrow{w}\nu$ and the integrand being continuous and bounded. The second inequality is from $h_{\delta,p-\eta+\delta}(v)\ge \mathbf{1}_{v\ge p-\eta+\delta}$ and $\forall v\in[p-\eta,p-\eta+\delta]$ $p>v\ge c(v)$. The third inequality is straightforward. The fourth inequality is from $\lim p_n> p-\eta+\delta$. The last inequality is from $\lim p_n<p+\eta$. 
	
	Letting $\eta\to 0$, we obtain $\pi(c,\nu,p)\ge \varlimsup \pi(c,\nu_n,p_n)$.
\end{proof}

\section{Proof of \autoref{prop:negative:frontier}}
\label{appendix:negative}
\begin{proof}
	We first show that $\Pset(\Gamma)$ is included in the set defined in \autoref{prop:negative:frontier}. $\forall (\pi_b,\pi_s)\in \Pset(\Gamma)$, it is clear that $\pi_b\ge0$ and $\pi_s\ge \piall(\Gamma)$. Now we prove that the inequality $\lambda\pi_b+\pi_s\le S_{\lambda}(\Gamma)$ is satisfied for any $\lambda\in[1,\infty)$. Let $\tau$ be the information structure and $(\sigma,\alpha,\nu)$ be an equilibrium  with payoff $(\pi_b,\pi_s)$. We define the following Borel measure $\beta$: for any Borel set $V'\subset \mathbb{R}$,
	\begin{align*}
		\beta(V')=\int_{v\in V'} \alpha(p,t_b)\sigma(\d p|t_s)P(\d t_b,\d t_s,\d v).
	\end{align*}
	In words, $\beta$ calculates the trading probability for a given set of types $V'$. By definition,
	\begin{align*}
	\pi_b&=\int (v-\E[p|v])\beta(\d v),\\
	\pi_s&=\int (\E[p|v]-c(v))\beta(\d v),
    \end{align*}
and hence
\begin{align*}
    \lambda\pi_b+\pi_s&=\int (\E[p|v]-c(v)+\lambda(v-\E[p|v]))\beta(\d v)\\
				&=\int(\lambda v-c(v)-(\lambda-1)\E[p|v])\beta(\d v)\\
				&\le \int (\lambda v-c(v)-(\lambda-1)\underline{v})\beta(\d v)\\
				&\le S_{\lambda}(\Gamma),
	\end{align*}
where the first inequality uses any on-path price being no lower than $\underline{v}$ and $\lambda\ge 1$.\par
Now we show that all payoff pairs $(\pi_b,\pi_s)$ satisfying the inequality constraints can be implemented by equilibrium payoffs in $\Pset(\Gamma)$. $\forall \lambda\in[1,\infty)$ and $\alpha\in[0,1]$, define $\beta_{\alpha}(v)=\mathbf{1}_{\lambda v+\underline{v}> c(v)+\lambda \underline{v}}+\alpha\mathbf{1}_{\lambda v+\underline{v}=c(v)+\lambda \underline{v}}$. Let us ignore the individual rationality constraint $\pi_s\ge 0$ for now. \par
Construct the following information structure: a public signal is sent to both players indicating whether $\beta_{\alpha}(v)=0$. Following the positive signal, construct an information structure as in \autoref{thm:joint} that induces Seller selling with probability one at $p=\underline{v}$ almost surely in the subgame.\footnote{In the subgame following $\beta_{\alpha}(v)>0$, the support of $v$ might not contain $\underline{v}$. This does not affect the consistency of off-path beliefs as we use wPBE as the equilibrium notion (without imposing subgame perfection).} When $\beta_{\alpha}(v)=0$, trading surplus is non-positive, as $c(v)\ge v+(\lambda-1)(v-\underline{v})$ and so an equilibrium with no trade exists. The realized weighted total surplus is exactly $S_{\lambda}(\Gamma)$, with Buyer's payoff $\int \beta_{\alpha}(v)(v-\underline{v})\mu(\d v)$ and  Seller's the remaining $\int \beta_{\alpha}(v)(\underline{v}-c(v))\mu(\d v)$. 
Note that Buyer's payoff is continuously increasing in $\alpha$. Therefore, $\forall \lambda\in[1,\infty)$, we can implement an interval (possibly degenerate) on the frontier $S_{\lambda}(\Gamma)$ defined by $\left\{ (\pi_b^{\alpha},\pi_s^{\alpha})= \left(\int \beta_{\alpha}(v)(v-\underline{v})\mu(\d v),\int \beta_{\alpha}(v)(\underline{v}-c(v))\mu(\d v)\right)\right\}_{\alpha\in[0,1]}$. Since we construct the equilibria explicitly, this interval satisfies all other constraints.\par
Next, we show two key properties of the interval $\left\{ (\pi_b^{\alpha},\pi_s^{\alpha}) \right\}_{\alpha\in[0,1]}$.
\begin{itemize}
	\item $\forall \lambda\ge 1$, $\forall \delta>0$, $(\pi_b',\pi_s')=(\pi_b^1+\delta,\pi_s^1-\lambda\delta)$ violates some frontier $S_{\lambda'}(\Gamma)$ with $\lambda'>\lambda$. It is straightforward to calculate
	\begin{align*}
		&\lambda'\pi_b'+\pi_s'=(\lambda'-\lambda)(\pi_b^1+\delta)+S_{\lambda}(\Gamma),
    \end{align*}
    and hence,
	\begin{align*}
\frac{\lambda'\pi_b'+\pi_s'-S_{\lambda}(\Gamma)}{\lambda'-\lambda}=(\pi_b^1+\delta).
	\end{align*}
	Now we calculate $S_{\lambda'}(\Gamma)$:
	\begin{align*}
		\frac{S_{\lambda'}(\Gamma)-S_{\lambda}(\Gamma)}{\lambda'-\lambda}=&\frac{\int_{\lambda(v-\underline{v})+\underline{v}-c(v)\ge 0}(\lambda'-\lambda)(v-\underline{v})\mu(\d v)}{\lambda'-\lambda}\\
																																			&+\frac{\int_{\lambda(v-\underline{v})+\underline{v}-c(v)\in[(\lambda'-\lambda)(\underline{v}-v),0)}(\lambda'(v-\underline{v})+\underline{v}-c(v))\mu(\d v)}{\lambda'-\lambda}
\\
\to&\int\beta_1(v)(v-\underline{v})\mu(\d v)=\pi_b^1\ \text{when }\lambda'\to \lambda.
\end{align*}
The limit is derived by canceling out $(\lambda'-\lambda)$ in the first line and observing the integrand is bounded by $(\lambda'-\lambda)(v-\underline{v})$ in the second line. Since $\delta>0$, we have that when $\lambda'-\lambda$ is sufficiently small, $\lambda'\pi_b'+\pi_s'>S_{\lambda'}(\Gamma)$, violating the frontier $S_{\lambda'}(\Gamma)$.
\item $\forall \lambda> 1$, $\forall \delta>0$, $(\pi_b',\pi_s')=(\pi_b^0-\delta,\pi_s^0+\lambda\delta)$ violates some frontier $S_{\lambda'}(\Gamma)$ with $\lambda'<\lambda$. The argument is symmetric.
\end{itemize}\par
Any point on the frontier $S_1(\Gamma)$ between $(0,S_1(\Gamma))$ and $(\pi_b^1,\pi_s^1)$ can be implemented by public randomization. Therefore, any payoff pair on the envelope of all frontiers is implementable (while ignoring Seller's individual rationality constraint). Then we can just truncate below by the extra constraint $\pi_s\ge 0$. The implementation of $(0,\piall(\Gamma))$ is trivial. Then public randomization implements all other points in the set. 
\end{proof}

\section{Proof of \autoref{thm:seller_linear}}
\label{appendix:seller:linear}
\begin{proof}
	As discussed in the main text, it is sufficient to show the existence of cdf $G(v)$ such that 1) $G\in D(\mu)$, 2) $G$ is an IPD, 3) $\exists p\in\Supp(G)$ s.t. $\int_{\ubar{v}}^pG(s)\d s=\int_{\ubar{v}}^pF(s)\d s$. \par
	First, we show that $\forall v_*\in [\ubar{v},\E_{\mu}[v]]$ s.t. $v_*>c(v_*)$ and $v_*\ge c(\E[v])$, IPD $G_{v_*}$ exists.\par
	\emph{Case 1}: $\lambda\neq 1$. The indifference condition of IPD is equivalent to:
	\begin{align*}
		&\frac{\d}{\d v}\int_v^{v^*}(v-c(s))\d G_{v_*}(s)=0\\
		\iff&-(v-c(v))g_{v_*}(v)+(1-G_{v_*}(v))=0\\
			\iff&(c(v)-v)\d \log\left( 1-G_{v_*}(v) \right)=1\\
				\iff& G_{v^*}(v)=1-C \left( c(v)-v \right)^{\frac{1}{\lambda-1}}.
	\end{align*}
	Using condition $G_{v_*}(v_*)=0$, we can pin down $C$:
	\begin{align*}
		G_{v_*}(v)=1-\left( \frac{v-c(v)}{v_*-c(v_*)} \right)^{\frac{1}{\lambda-1}}.
	\end{align*}
	Lastly, $v^*$ can be pinned down using the following condition:\footnote{It is easy to verify that the condition is equivalent to $\int v\d G_{v_*}(v)=\E_{\mu}[v]$.}
	\begin{align*}
		&(1-G_{v_*}(v^*))(v^*-c(v^*))=v_*-c(\E_{\mu}[v])\\
		\implies&(1-\lambda)v^*-\gamma=\left( v_*-c(\E[v]) \right)^{\frac{\lambda-1}{\lambda}}\left(  (1-\lambda)v_*-\gamma \right)^{\frac{1}{\lambda}}.
	\end{align*}
	Note that if $v_*\to\E[v]$, then $v^*\to\E[v]$. One can also verify that $v^*$ increases when $v_*$ decreases:
	\begin{align*}
		\frac{\d v^*}{\d v_*}=(v_*-\E[v])(v_*-c(\E[v]))^{-\frac{1}{\lambda}}( (1-\lambda)v_*-\gamma)^{\frac{1-\lambda}{\lambda}}\le 0.
	\end{align*}\par
	\emph{Case 2}: $\lambda=1$ (hence $\gamma<0$). The indifference condition of IPD is equivalent to:
\begin{align*}
	&\frac{\d}{\d v} \log\left(1- G_{v_*}(v) \right)=\frac{1}{\gamma}\\
	\iff& G_{v_*}(v)=1-C\cdot e^{\frac{v}{\gamma}}.
\end{align*}
We pin down $C$ and $v^*$ using the CDF at $v_*$ and the mean-preserving-spread condition:
\begin{align*}
	\begin{dcases}
		C=e^{-\frac{v_*}{\gamma}}\\
		v^*=v_*+\gamma\log\left( \frac{v_*-c(\E_{\mu}[v])}{-\gamma} \right).
	\end{dcases}
\end{align*}\par

Second, we show that there exists $v_*$ s.t the corresponding IPD $G_{v_*}(v)$ satisfies condition 1) and 3). It is trivial that if $v_*=\E_{\mu}[v]$, then $G_{v_*}$ has unit mass at $\E_{\mu}[v]$, which is included in $D(\mu)$. By the linearity of $c$, $v=c(v)$ can only happen on the boundaries of $V$. Therefore, $v_*=\inf\left\{ v'\in(\max\{\ubar{v},c(\E[v])\},\E_{\mu}[v]]\big| G_{v'}\in D(\mu) \right\}$ is well defined. We now consider three cases separately:
\begin{itemize}
	\item \emph{Case 1}: $v_*>\max\{c(v_*),c(\E[v])\}$. In this case $\nu_{v_*}$ is well defined. By the formula of $G_{v'}(v)$, it is continuous in $v'$ for each $v$. Obviously, CDFs are uniformly bounded. So by dominated convergence theorem, $\forall q$:
		\begin{align*}
			\int_{\ubar{v}}^{v}G_{v_*}(s)\d s=\lim_{v'\to v_*^+}\int_{\underline{v}}^vG_{v'}(s)\d s.
		\end{align*}
		This implies $G_{v_*}\in D(\mu)$. Now we claim that there exists $p\in [v_*,v^*]$ such that $\int_{\ubar{v}}^p F(s)\d s=\int_{\ubar{v}}^p G_{v_*}(d)\d s$. If not, this implies $\int_{\ubar{v}}^p F(s)\d s<\int_{\ubar{v}}^p G_{v_*}(d)\d s$ $\forall p\in[v_*,v^*]$. Then choosing $v_*$ slightly smaller, $G_{v_*}$ is still in $D(\mu)$, contradiction.
	\item \emph{Case 2}: $v_*=c(v_*)\ge c(\E[v])$. We show that this case is never possible. Since $c$ is linear, this can happen only when $v_*=\ubar{v}$ and $\lambda\le0$. Consider $v_*'=\ubar{v}+\epsilon$ where $\epsilon>0$. Then $G_{v_*'} \in D(\mu)$ when $\varepsilon$ is very small. However, $v^*$ is pinned down by:
		\begin{align*}
			(1-\lambda)v^*=\gamma+(v_*'-c(\E[v]))^{\frac{\lambda-1}{\lambda}}( (1-\lambda)\epsilon )^{\frac{1}{\lambda}}.
		\end{align*}
		When $\epsilon\to 0$, $v^*\to \infty$, so $G_{v_*'}\not\in D(\mu)$ for sufficiently small $\epsilon$, contradiction.
	\item \emph{Case 3}: $v_*=c(\E[v])>\ubar{v}$. In this case, $G_{v_*}$ gives Seller zero profit and $G_{v^*}\in D(\mu)$. Then $\pius(\Gamma)=0$. So the proof of \autoref{thm:seller_linear} is already done. \qedhere
\end{itemize}
\end{proof}

\section{Proof of \autoref{coro:binary}}
\label{appendix:coro:binary}
\begin{proof}
	It straightforward to verify that the solution to \autoref{eqn:binary} is unique; denote it by $p$.
	When $V$ is binary, $c(v)$ is trivially affine. Let $\lambda=\frac{c(v_2)-c(v_1)}{v_2-v_1}$ and $\gamma=c(v_i)-\lambda v_i$. Then, \autoref{ass:linear} is satisfied and \autoref{thm:seller_linear} applies. Let $v_*$ be the corresponding parameter defining profit minimizing IPD. Since $V$ is binary, $\int_{\ubar{v}}^vF(s)\d s$ is a piecewise linear function with two kinks at $v_1,v_2$. Meanwhile, $\int_{\ubar{v}}^vG_{v_*}(s)\d s$ is strictly convex on its support $(v_*,v^*)$. Therefore, $\int_{\ubar{v}}^vF(s)\d s$ can not intersect $\int_{\ubar{v}}^vG_{v_*}(s)\d s$ at any $v\in (v_*,v^*)$. So either ($v_*=v_1$ and $v^*\le v_2$) or ($v^*=v_2$ and $v_*\ge v_1$). 
    
    We begin with the conjecture that $v^*=v_2$. This implies:
	\begin{align*}
		&(1-G_{v_*}(v_2))(v_2-c(v_2))=v_*-\E_{\mu}[c(v)]\\
		\implies&		(v_*-c(v_*))^{\frac{1}{\lambda-1}}(v_*-\E_{\mu}[c(v)])=(v_2-c(v_2))^{\frac{\lambda}{\lambda-1}},
	\end{align*}
    i.e., $v_*$ solves \autoref{eqn:binary} and $v_*=p$. Therefore, only when $p\ge v_1$ the conjecture is valid, in which case $p$ is an optimal price and $\pius(\Gamma)=p-\E[c(v)]$. 
    
    Otherwise, if $p<v_1$, the conjecture $v^*=v_2$ is not valid, so $v_1$ is an optimal price and $\pius(\Gamma)=v_1-\E[c(v)]$. To sum up, $\pius(\Gamma)=\max\{p,v_1\}-\E[c(v)]$.
\end{proof}

\bibliographystyle{ecta}
\bibliography{KZ}

\end{document}